\documentclass[12pt]{article}
\pdfoutput=1

\usepackage{formatting}

\usepackage{shortcuts}

\begin{document} 
\begin{titlepage}
\begin{flushright}BRX-TH-6718\end{flushright}
\begin{center}
\phantom{ }
%\vspace{1.5cm}

{\bf \large{Geometric Surprises in the Python's Lunch Conjecture}}
\vskip 0.5cm
{Gurbir Arora${}^1$, Matthew Headrick${}^1$, Albion Lawrence${}^{1,2}$, Martin Sasieta${}^1$,  Connor Wolfe${}^1$}
\vskip 0.05in
\small{  \textit{${}^1$ Martin Fisher School of Physics, Brandeis University}}
\vskip -.4cm
\small{\textit{Waltham, Massachusetts 02453, USA}}
\vskip -.4cm
\small{  \textit{${}^2$ California Institute of Technology, Pasadena, California 91125, USA}}

\vskip -.10cm

\begin{abstract}
A \emph{bulge surface}, on a time reflection-symmetric Cauchy slice of a holographic spacetime, is a non-minimal extremal surface that occurs between two locally minimal surfaces homologous to a given boundary region. According to the python's lunch conjecture of Brown et al., the bulge's area controls the complexity of bulk reconstruction, in the sense of the amount of post-selection that needs to be overcome for the reconstruction of the entanglement wedge beyond the outermost extremal surface. We study the geometry of bulges in a variety of classical spacetimes, and discover a number of surprising features that distinguish them from more familiar extremal surfaces such as Ryu-Takayanagi surfaces: they spontaneously break spatial isometries, both continuous and discrete; they are sensitive to the choice of boundary infrared regulator; they can self-intersect; and they probe entanglement shadows, orbifold singularities, and compact spaces such as the sphere in AdS$_p\times S^q$. These features imply, according to the python's lunch conjecture, novel qualitative differences between complexity and entanglement in the holographic context. We also find, surprisingly, that extended black brane interiors have a non-extensive complexity; similarly, for multi-boundary wormhole states, the complexity pleateaus after a certain number of boundaries have been included. 
\end{abstract}
\end{center}

\small{\,\\[.01cm]
\href{mailto:gurbir@brandeis.edu}{gurbir@brandeis.edu} \\
\href{mailto:mph@brandeis.edu}{mph@brandeis.edu} \\
\href{mailto:albion@brandeis.edu}{albion@brandeis.edu}\\
\href{mailto:martinsasieta@brandeis.edu}{martinsasieta@brandeis.edu}\\
\href{mailto:cwolfe@brandeis.edu}{cwolfe@brandeis.edu} 
}

\end{titlepage}

\setcounter{tocdepth}{3}

\hrule 
{\parskip = .2\baselineskip \tableofcontents}

 \vskip .5cm
 \hrule

\section{Introduction}

A central theme of the past decade and a half of work in holography has been that the emergence of the bulk space is reflected in quantum information-theoretic properties of the boundary system. The foundational result is the Ryu-Takayanagi (RT) formula for the spatial entanglement of the holographic state, which in its quantum-corrected form is \cite{Ryu:2006bv,Faulkner:2013ana}:
\be\label{eq:qcRT}
 S(\rho_{\R}) \approx  \dfrac{\text{Area}(X_{\R})}{4G} + S(\rho_{\rew})\;.
\ee 
Here the holographic state is assumed to admit a semiclassical bulk description and to be time reflection-symmetric. The semiclassical state lives at the time reflection-symmetric bulk Cauchy slice $\Sigma$, $\R$ is a subregion of the conformal boundary $\partial \Sigma$ of this slice, $\rho_{\R}$ is the reduced density matrix of the boundary quantum field theory to this subregion, and $S(\rho_{\R})$ is the corresponding von Neumann entropy. On the right-hand side of this equation, $X_{\R}$ is the minimal area surface on $\Sigma$ which is homologous to $\R$, $\rho_{\rew}$ is the reduced density matrix of the bulk quantum fields on the region bounded by $\R$ and $X_{\R}$, and $S(\rho_{\rew})$ is its von Neumann entropy. The right-hand side of this equation defines the {\it generalized entropy} of $X_{\R}$ in the corresponding bulk state, denoted by $S_{\text{gen}}(X_{\R})$. This geometric dictionary has provided a great deal of insight into the emergence of bulk space from the spatial entanglement structure of the boundary quantum field theory \cite{Ryu:2006bv,Swingle:2009bg,VanRaamsdonk:2010pw,Czech:2012bh}.

In this correspondence, the boundary state $\rho_{\R}$ is conjectured to contain all of the quantum information about the state $\rho_{\rew}$; and the bulk physics in the {\it entanglement wedge} of $\R$, the bulk domain of dependence ${\cal D}_{\rew}$, is conjectured to be holographically represented in the physics of the boundary domain of dependence ${\cal D}_{\R}$ \cite{Czech:2012bh,Wall:2012uf, Headrick:2014cta}. This statement is known as ``entanglement wedge reconstruction'' and has been given a more precise formulation within the past several years, in the language of quantum information theory. In this formulation, one first considers the bulk-to-boundary map $V: \mathcal{H}_{\Sigma}\rightarrow \mathcal{H}_{\text{CFT}}$, where $\mathcal{H}_{\Sigma}$ is the Hilbert space of bulk effective field theory excitations on $\Sigma$, commonly dubbed the {\it code subspace}, and $\mathcal{H}_{\text{CFT}}$ is the CFT Hilbert space defined on $\partial \Sigma$. This map consequently defines the encoding channel onto the subregion $\R$ of the conformal boundary:\footnote{For ease of exposition, we are assuming that all of the Hilbert spaces factorize. Non-factorizability can be treated in the language of von Neumann algebras.}
\be
\mathcal{N}(\rho_{\rew}) = \text{Tr}_{\bar{\R}} (V\rho_\rew \otimes \sigma_{\bar{\rew}} V^\dagger) = \rho_{\R}\,.
\ee
Here $\sigma_{\bar{a}}$ can be any reference full-rank density matrix on the bulk complement region $\bar{a}$. The statement of entanglement wedge reconstruction is that $\mathcal{N}$ admits a recovery map $\mathcal{R}$, i.e.\ an inverse channel satisfying $\mathcal{R}\circ \mathcal{N}(\rho_{\rew}) \approx \rho_{\rew}$. The existence of such a recovery channel $\mathcal{R}$ has been identified as an inevitable consequence of \eqref{eq:qcRT} holding within states of the code subspace \cite{Jafferis:2015del,Dong:2016eik,Harlow:2016vwg,ohya1993quantum,Cotler:2017erl}.

Moreover, the structure of entanglement wedges in hyperbolic space is very suggestive, given that there are spatial regions $\mathsf{b}$ localized deep in AdS that are contained in the entanglement wedge of the union of disjoint boundary subregions, $\R = \sqcup_i \R_i$, while not being contained on any of the individual entanglement wedges, $\mathsf{b} \cap  \rew_i = \emptyset$. In this sense, the way the bulk quantum information is distributed on the boundary is reminiscent of (operator algebra) quantum error correcting codes, where the quantum information of the logical system is distributed non-locally on the physical system \cite{Swingle:2009bg,Almheiri:2014lwa,Dong:2016eik,Harlow:2016vwg}. Special ``holographic codes'' modelling these features of $V$ can be designed in qubit systems \cite{Pastawski:2015qua,Hayden:2016cfa}, using tensor networks that provide a discretization of the emergent hyperbolic space (see \cite{Bao:2018pvs} for more realistic constructions). Tensor network models can also be used to qualitatively describe the map $V$ for time-reflection symmetric but otherwise general reference states, with associated bulk time reflection-symmetric Cauchy slices $\Sigma$.

More generally, entanglement wedge reconstruction is expected to work in situations that include non-trivial dynamical evolution of the bulk geometry, in which case $X_\R$ is replaced by the HRT surface \cite{Hubeny:2007xt} in \eqref{eq:qcRT}, or, moreover, in situations where the bulk entanglement $S(\rho_{\rew})$ becomes $O(G_N^{-1})$, in which case \eqref{eq:qcRT} is replaced by the full-fledged quantum extremal surface (QES) prescription \cite{Engelhardt:2014gca}. These extensions provide a new perspective on the way in which the information escapes from an evaporating black hole \cite{Penington:2019npb,Almheiri:2019psf}. Namely, the possibility of recovering the information from the Hawking radiation is manifested semiclassically after the Page time in terms of a non-trivial minimal QES  delimiting an ``island'' in the entanglement wedge of the radiation \cite{Penington:2019npb,Almheiri:2019psf,Almheiri:2019hni}.

In these situations, the statement of entanglement wedge reconstruction becomes particularly surprising, since the entanglement wedge ${\cal D}_{\rew}$ generally contains regions that are causally inaccessible from $\R$. That is, the {\it causal wedge} of $\R$, the set of bulk points that are both in the future and in the past of ${\cal D}_{\R}$, is strictly contained in the entanglement wedge. The causal wedge is accessible through correlators of appropriately smeared local operators on the boundary through the bulk-to-boundary operator map \cite{Banks:1998dd, Balasubramanian:1998de,Balasubramanian:1999ri,Hamilton:2006az}. In fact, larger regions have been identified to be accessible in a simple way, by being able to manipulate simple sources in the boundary Hamiltonian \cite{Engelhardt:2021mue}. These regions are delimited by the apparent horizon, or more precisely, by the outermost QES. More complicated operators are expected to be needed to access the region of the entanglement wedge that lies beyond this region. Characterizing these operators is a major open problem in holography and lies at the core of the black hole information problem.

\subsection{Python's lunch conjectures}

\begin{figure}[h]
 		\centering
 		\includegraphics[width = .5\textwidth]{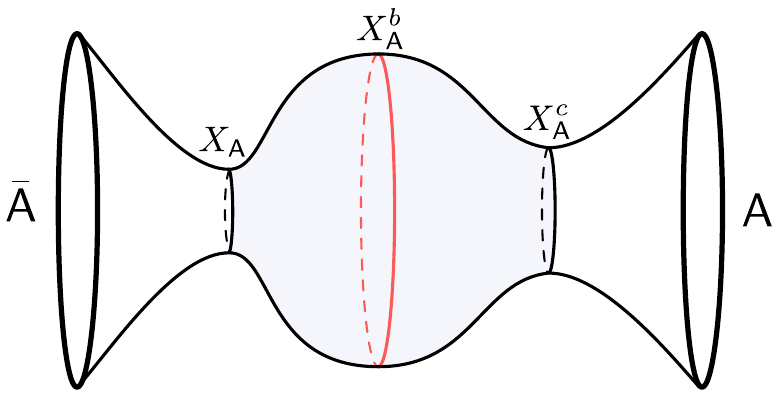}
 		\caption{Structure of slice $\Sigma$ containing a python: the globally minimal QES $X_\R$, the bulge $X_\R^b$ and the constriction $X^c_\R$. The python, shaded in blue, is delimited by $X_\R$ and $X_\R^c$. The entanglement wedge of $\R$ contains the python. { Throughout this paper, bulge surfaces will be consistently drawn in red.} }
 		\label{fig:PL}
 \end{figure}

A step toward understanding bulk reconstruction beyond the outermost QES was given in \cite{Brown:2019rox}, based on the analogy with tensor network toy models of the bulk-to-boundary map $V$. In the simplest case, illustrated in Fig.\ \ref{fig:PL}, the bulk Cauchy slice $\Sigma$ contains two locally minimal QESs: the outermost QES, $X^c_\R$, called the {\it constriction}, and the globally minimal QES $X_\R$, which delimits the entanglement wedge $\rew$. The region between them is called a \emph{python's lunch}, or just \emph{python} for short. In the python, there exists a third QES homologous to $\R$ that is not locally minimal, $X_\R^b$, called the \textit{bulge}. In such a situation, the python's lunch conjecture (PLC) \cite{Brown:2019rox} { (see also \cite{Engelhardt:2021mue,Engelhardt:2021qjs})} assigns a unitary complexity to $\mathcal{R}$ given at leading order by 
\be\label{eq:PLC}
\Co(\mathcal{R}) \sim \exp\left(\dfrac{S_{\text{gen}}(X_\R^b)-S_{\text{gen}}(X^c_\R)}{2}\right),
\ee
where we are omitting subexponential volume factors that become unimportant in most situations.

The motivation behind this proposal comes from tensor-network toy models of $V$, where the geometry of $\Sigma$ gets discretized in the form of a graph with local tensors at the vertices. The task of reconstructing the entanglement wedge is to undo the part of the network representing $\rew$, by acting with local unitaries on $\R$. Consequently, the complexity to perform this operation gets ``geometrized'' by the structure of the network. Namely, the reconstruction requires one to undo the network locally, and to do that one needs to go from the locally minimal cut representing $X^c_\R$ to the locally maximal cut representing $X_\R^b$. Each of these cuts defines an auxiliary Hilbert space, and the part of the network between the cuts defines an isometric map between these Hilbert spaces. By construction of the network, the log bond dimension of these cuts is $S_{\text{gen}}(X^c_\R)$ and $S_{\text{gen}}(X_\R^b)$ respectively. Since $S_{\text{gen}}(X_\R^b)>S_{\text{gen}}(X^c_\R)$, the inverse of this map post-selects (i.e.\ orthogonally projects out) a Hilbert subspace of log dimension $S_{\text{gen}}(X^b_\R) - S_{\text{gen}}(X^c_\R)$ of the $X_\R^b$ cut. In quantum information, under genericity assumptions for the local gates, the optimal way to overcome this post-selection unitarily is to introduce ancilla qubits and perform a brute force Grover search, at the cost of an exponentially large number of few-body unitary operations, parametrically given by $\exp\frac12(S_{\text{gen}}(X^b_\R) - S_{\text{gen}}(X^c_\R))$. This motivates the particular form of the exponent in \eqref{eq:PLC}, which is given in terms of generalized entropies and can be directly translated to the semiclassical description of $\Sigma$ in AdS/CFT.

The PLC \eqref{eq:PLC} is able to describe why, in some specific situations, the global bulk-to-boundary map $V$ remains ``simple'', with unitary complexity scaling polynomially with some extensive parameter of the boundary, like its thermodynamic entropy $S$, while, at the same time, the complexity to reconstruct the entanglement wedge $\rew$ for any proper subsystem $\R$ of the boundary is exponentially large in $S$. This dichotomy is what originally motivated the conjecture in \cite{Brown:2019rox}. Examples of these situations arise when $V$ is constructed dynamically from the unitary time-evolution operator of the system, $V= \exp(-iHt)$, driven by a chaotic few-body Hamiltonian $H$ that couples $\R$ and $\bar{\R}$. Any initial information scrambles rapidly throughout the system, which generically opens the possibility of recovering it from any subsystem $\R$ containing more than half of the entropy of the full system.  On the one hand, with access to the full system, undoing the time-evolution $V$ is by assumption ``simple'' for subexponential timescales. On the other hand, the recovery from $\R$ typically requires exponentially many unitary operations, with $\Co(\mathcal{R})  \sim e^{ (S- S_\R)/2}$, from the fact that there is a python for $\R$.  For specific models of black hole evaporation, these features of the bulk-to-boundary map were studied in \cite{Hayden:2007cs,Harlow:2013tf}.

\begin{figure}[h]
 		\centering
 		\includegraphics[width = .7\textwidth]{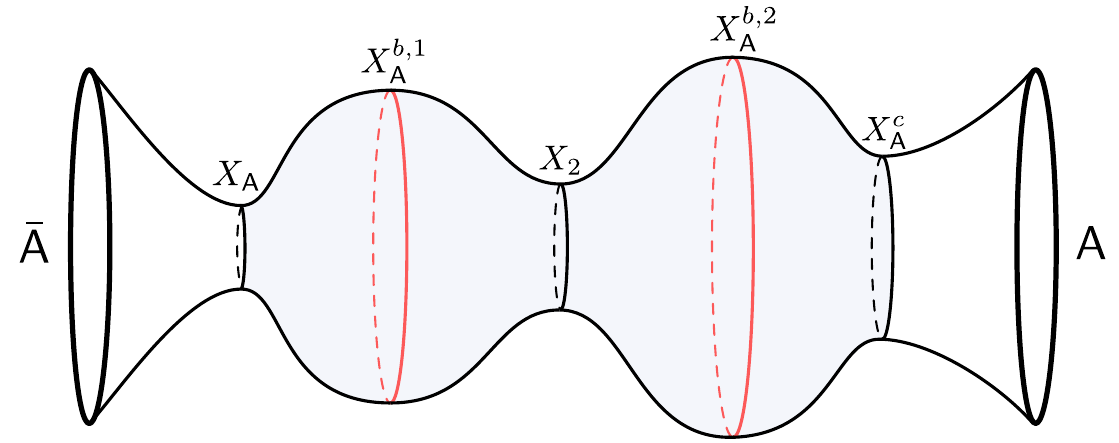}
 		\caption{Schematic structure of a python with multiple lunches:   between the globally minimal QES $X_\R$ and the constriction $X^c_\R$, there are  locally minimal QESs such as $X_2$.}
 		\label{fig:PLmulti}
 \end{figure}
 
A generalization of the PLC to situations with multiple ``lunches'' on the same Cauchy slice was given in \cite{Engelhardt:2021qjs}. We represent the situation with two lunches in Fig.\ \ref{fig:PLmulti}. In general, one considers $n$ lunches, defined by a set of non-intersecting locally minimal homologous QESs $\mathcal{S} = \lbrace X_1,X_2,...,X_n\rbrace $ between the globally minimal QES $X_1 \equiv X_\R$ and the constriction $X_n \equiv X_\R^c$. In between each pair of adjacent minimal QESs $X_i$ and $X_{i+1}$, there will be a bulge, $X^{b,i}_\R$, for $i=1,...,n-1$. According to the generalized PLC, the complexity to reconstruct the lunch from $\R$ is 
\be\label{eq:PLC2}
\Co(\mathcal{R}) \sim \max\limits_{i<j} \left\lbrace  \exp\left( \dfrac{S_{\text{gen}}(X_j^b)-S_{\text{gen}}(X_i)}{2} \right) \right\rbrace \;.
\ee 
where $i<j$ means that the maximization is restricted to minimal QESs, labelled by $i$, that lie between the bulge, labelled by $j$, and $\R$. Intuitively, this expression simply represents that the total amount of post-selection is largely dominated by the maximum generalized entropy difference between any bulge and any minimal cut that lies closer to $\R$.

\subsection{This paper}

The precise meaning of the complexity in Eqs. \eqref{eq:PLC} and \eqref{eq:PLC2} is to date incomplete. The goal of this paper is to provide additional data for sharpening these conjecture by more fully exploring the properties of the bulge surface $X^b_\R$ for the simplest species of python: geometric pythons that arise already at the level of the classical bulk geometry dual to some reference holographic state. For these pythons, all the different surfaces $X_\R, X^b_\R$, $X^c_\R$ will be extremal area surfaces, and their generalized entropies will all be given by the area term in \eqref{eq:qcRT}, to leading order in the semiclassical expansion. 

Furthermore, in this paper, we will work on a fixed partial Cauchy slice $\Sigma$ of the bulk spacetime and consider extremal surfaces for variations restricted to $\Sigma$. For time reflection-symmetric states of the holographic system, we expect that, generally, the natural choice of partial Cauchy slice consists of the bulk moment of time symmetry $\Sigma$. For instance, partial maximin surfaces, like the constriction or the RT surface, must lie on $\Sigma$. Therefore, our definition holds for the states for which this is true. However, it is important to note that, even restricted to time-reflection symmetric states, this is not always true. Explicit states for which time reflection symmetry is spontaneously broken by the bulge and other minimal surfaces have been constructed recently, in near-extremal black hole interiors in \cite{Engelhardt:2023bpv}.\footnote{{ Other states with time-reflection symmetric pythons in JT gravity and massless matter were constructed in \cite{Bak:2021qbo}.}} We do not believe that this phenomenon occurs in any of the spacetimes we study in this paper, although we have not proven that this is the case.

In this context, we will point out several properties that $X^b_\R$ fails to satisfy that are usually taken for granted for the locally minimal surfaces $X^c_\R$ and $X_\R$. Assuming these properties for a general extremal surface leads to the incorrect identification of the bulge in many situations. We will show that this has important physical consequences for the outcome of the python's lunch conjectures \eqref{eq:PLC}, \eqref{eq:PLC2}. Throughout the paper, we will identify the true bulge surface in different situations, study its topological and geometric properties, and comment on the outcome of the conjecture given these properties.

The paper is organized as follows. We begin in section \ref{sec:minmax}\ by studying the implications of the minimax definition of the bulge surface given in \cite{Brown:2019rox}, and relating it to a branch of geometric measure theory called Almgren-Pitts min-max theory, developed by mathematicians for the purpose of proving the existence and properties of extremal surfaces. We also review the motivation of the bulge as a minimax surface from tensor network heuristic models. With these results in hand we proceed to explore a number of different examples of holographic states with classical pythons:
\begin{itemize}
\item In section \ref{sec:breakisom}, we demonstrate that, unlike minimal surfaces, bulges can break continuous and discrete spatial isometries of the slice $\Sigma$ and region $\R$. This effect has an immediate consequence for generic states of black branes with semiclassical interiors, namely, according to \eqref{eq:PLC}, the interiors are simple to reconstruct; specifically, the log-complexity is not extensive in the boundary volume. Furthermore, in the planar limit, the bulge, and hence according to the conjecture the complexity, are highly sensitive to the choice of infrared regulator.

\item In section \ref{sec:vacuumpython}, we explore examples of bulges that arise in vacuum anti-de Sitter (AdS) space. We start by exploring some simple examples arising when $\R$ is comprised of disconnected boundary subregions in AdS$_3$ and AdS$_4${, including, for AdS$_3$, where $\R$ covers the entire boundary except a discrete set of points. This includes an example where a python is present, yet the exponent in the complexity vanishes.} We then show that, in the presence of compact extra dimensions, even when the metric of $\Sigma$ is of product form, $X_\R^b$ is generically {\it not} of product form. The generalized entropy of the bulge $S_{\text{gen}}(X_\R^b)$ thus contains dynamical information about the holographic system that goes beyond the spatial correlations in the ground state of the holographic CFT. We find that this effect resolves the singular bulges previously found in AdS$_3$.

\item In section \ref{sec:excitedstates}, we describe bulges on excited states of the holographic CFT{. We include examples in which the dual geometry has no horizon, namely AdS$_3$ orbifolds and Lin-Lunin-Maldacena (LLM) geometries. Additionally, we describe bulges that form outside of the horizon of an eternal black hole.}

\item In section \ref{sec:bhint}, we explore microstates of multiple black holes with pythons occupying their shared semiclassical interiors. For two-sided states, we show that the complexity to reconstruct the interior with access to $\R$ is the same, whether $\R$ contains both boundaries or just the  boundary that can access the interior. We show that this is a general feature of multi-boundary wormhole microstates, namely, the complexity to reconstruct the interior plateaus after a number of boundaries has been included in $\R$. This effect is essentially a discrete analogue of the non-extensivity of the log-complexity for black branes mentioned above. We also provide a slight generalization of the second python's lunch conjecture \eqref{eq:PLC2} in section \ref{sec:bhint}, where we point out that different choices might exist for the set $\mathcal{S}$ of non-intersecting minimal surfaces defining the lunch. Different choices of $\mathcal{S}$ cannot all be included on the same foliation, and therefore we must minimize the complexity \eqref{eq:PLC2} over these choices.
\end{itemize}
We regard these examples as providing data against which to check the PLC, in the hope that the complexity can be directly evaluated, or its properties studied, in the corresponding situations. We close with a summary and discussion in section \ref{sec:conclusions}. Some technical details concerning extremal surfaces in $\mathbf{R}^3$ are presented in the appendices.

\section{Mathematical background}
\label{sec:minmax}

In \cite{Brown:2019rox}, a \emph{bulge surface} in a holographic spacetime was defined mathematically via a certain maximinimax formula. The ``minimax'' part of the formula referred to operations on a Cauchy slice, and the ``maxi'' to a maximization over Cauchy slices, rendering the formula covariant. In this paper, we are focusing on surfaces lying on a constant-time slice of a static spacetime, or more generally lying on the $t=0$ slice of a time-reflection symmetric spacetime. Since we are fixing a Cauchy slice, we will focus on the ``minimax'' part of the formula, and put aside the ``maxi'' part.

In this section, we will review the minimax formula of \cite{Brown:2019rox}, and argue that bulge surfaces obeys certain properties that we will make use of in the rest of the paper. Our discussion will be far from mathematically rigorous. However, we will also point out that the minimax formula fits naturally within an existing body of mathematical work called ``Almgren-Pitts min-max theory''. This theory, a branch of geometric measure theory, has been developed since the 1960s as a set of techniques for proving the existence and properties of extremal submanifolds in general Riemannian manifolds. We will give a very brief sketch of some of the ideas in this theory, not because we will make use of them in this paper, but for completeness and to reassure the reader that this work can be put on a rigorous foundation if desired. (See \cite{colding2003min,marques2013applications} and references therein for additional details.)

\subsection{Extremal surfaces \& Morse index}

We begin by recalling some basic facts about extremal hypersurfaces in Riemannian manifolds. Let $N$ be a compact Riemannian manifold, possibly with boundary. We will denote coordinates on $N$ by $x^\mu$ and its metric by $g_{\mu\nu}$. By a \emph{surface} $X$ we mean a compact orientable hypersurface in $N$ such that $\partial X\subset\partial N$ and $\intt X\subset\intt N$ (where $\intt$ denotes the interior). We will denote coordinates on $X$ by $y^a$, the induced metric by $h_{ab}$, a continuous unit normal vector field by $n^\mu$, the extrinsic curvature (defined with respect to $n^\mu$) by $K_{ab}$, and its trace by $K$.

We now want to study variations in the area of $X$ under small deformations. Let $\eta$ be a smooth function on $X$ that vanishes on $\partial X$.\footnote{\label{foot:neumann}In subsection \ref{sec:planar}, we will also consider imposing a Neumann boundary condition on $\eta$, where the surface meets an end-of-the-world brane. The formulas in this subsection remain correct, including the self-adjointness of the Jacobi operator, with this boundary condition.} The first variation of the area under deforming $X$ by applying the exponential map to each point of $X$ by the vector $\epsilon\eta n^\mu$, is
\be
\delta\area(X)=\int_X{\text{d}}y\sqrt{h}\,K\eta\,.
\ee
The variation is therefore zero for any function $\eta$ if and only if $K=0$ everywhere. Again following physicists' conventions, we will call such a surface  \emph{extremal}. (Mathematicians use the term \emph{minimal}.)

Assume $X$ is extremal. The second variation of its area is:
\be\label{eq:secondvariation}
\delta^2\area(X)
=\frac{1}2\int_X{\text{d}}y\sqrt{h}\left[h^{ab}\partial_a\eta\partial_b\eta-(R_{\mu\nu}n^\mu n^\nu+K_{ab}K^{ab})\eta^2\right] =\frac{1}2\int_X{\text{d}}y\sqrt{h}\,\eta J\eta\;,
\ee
where $J$ is the following Schr\"odinger-type operator on $X$, called the \emph{Jacobi operator}:
\be\label{eq:Jacobi}
J:= -\nabla^2-R_{\mu\nu}n^\mu n^\nu-K_{ab}K^{ab}\;,
\ee
with $\nabla^2$ the Laplacian with respect to $h_{ab}$. From its definition, the Jacobi operator's spectrum is discrete (since $X$ is compact), bounded below, and unbounded above. The number of negative eigenvalues is called the (Morse) index of $X$, for the following reason. If we consider the space of all surfaces in $N$ with boundary equal to $\partial X$, the area functional defines a Morse (or Morse-Bott) function\footnote{A \emph{Morse} function on a manifold is a $C^2$ real function whose Hessian, at each critical point, is non-degenerate. The Morse index of a critical point $x$ is the largest dimension of a subspace of the tangent space $T_x$ on which the Hessian is negative definite. A useful generalization is a \emph{Morse-Bott} function, a function whose critical points form submanifolds on which the normal Hessian is nondegenerate. Writing the Hessian as $\langle\cdot,J\cdot\rangle$, where $\langle\cdot,\cdot\rangle$ is a positive-definite inner product on $T_x$ and $J$ is a symmetric operator, the Morse index equals the number of negative eigenvalues of $J$.}, of which $X$ is a critical point; by \eqref{eq:secondvariation}, the number of negative eigenvalues of $J$ is then equal to the Morse index. (More precisely, for a generic metric on $N$, the area functional is a Morse function; and for a metric that is generic subject to some isometry group, the area functional is a Morse-Bott function.)

The following facts about the index will be relevant to us in the rest of the paper. First, if the surface is a local minimum of the area, then the index vanishes; for a generic metric (or generic subject to some isometries), the converse holds. Second, suppose that $X$ is a disjoint union, $X=X_1\sqcup X_2\sqcup\cdots$. Then a basis of eigenfunctions of $J$ can be chosen such that each eigenfunction vanishes on all but one $X_i$; therefore the index of $X$ is simply the sum of the indices of the $X_i$. Third, suppose that $X$ is connected and has index 1. It is a standard fact from quantum mechanics that the ground state wave function of a Schr\"odinger operator has no nodes; therefore the eigenfunction $\eta$ has constant sign, in other words the unstable mode moves all of $X$ in the same direction.

\subsection{Min-max theory}

Now suppose that a subset $\tilde N$ of $N$ is bounded by two locally minimal surfaces $X_0,X_1$; keeping track of orientations, we have
\be
\partial\tilde N = X_0-X_1\,.
\ee
Necessarily then $X_0$ and $X_1$ are homologous, and share the same boundary: $\partial X_0=\partial X_1$. (More precisely, they are homologous relative to that boundary.\footnote{Here we work in homology relative to the codimension-2 boundary submanifold $\partial X_0$. In subsection \ref{sec:planar}, we will consider a slightly more general situation, in which we work in homology relative to a codimension-1 part of the boundary of $\tilde N$, representing an end-of-the-world brane, and we require surfaces to end orthogonally on this boundary, leading to the Neumann boundary condition mentioned in footnote \ref{foot:neumann}. As far as we know, this extra boundary does not affect the considerations of this section.}) In the holographic setting, $\tilde N$ is a ``python''. We will denote by $\tilde{H}$ the homology class of $X_{0,1}$. (One or both of $X_{0,1}$ may be empty, in which case the elements of $\tilde{H}$ are null-homologous.) By Almgren-Pitts min-max theory, there exists a third extremal surface with index 1 in $\tilde{H}$.

The argument rests on the \emph{mountain-pass lemma}, which is a general statement concerning Morse functions that guarantees the existence of an index-1 critical point, given two local minima $x_{0,1}$ connected by a path $\bar x(t)$.\footnote{Either the index-0 critical points must be distinct (or, in the Morse-Bott setting, lie on distinct connected components of the critical submanifold), or the path $\bar x(t)$ must be non-contractible.} This is proved by considering the following minimax problem:
\be
\min_{x(t)}\max_{t}f(x(t))\;,
\ee
where $f$ is a Morse function, the minimum is over paths $x(t)$ homotopic to $\bar x(t)$, and the maximum is over points on $x(t)$. The solution to this problem is an index-1 critical point.

In the case at hand, we are working in the homology class $\tilde{H}$.\footnote{More accurately, in geometric measure theory one works with a generalization of the notion of submanifold called a \emph{varifold}. The space of varifolds admits a natural topology that allows splitting, joining, and other degenerations as continuous processes. The burden is then to show that the varifold returned by the mountain-pass lemma is, under certain conditions actually a submanifold. In fact, as found in \cite{Brown:2019rox}, there are generic situations where the minimax surface is \emph{not} a submanifold; we will return to this example in subsection \ref{sec:vacuumads3}.} The initial path $\bar X(t)$ is given by the level sets of a Morse function $\psi$ on $\tilde N$ that equals 0 on $X_0$ and 1 on $X_1$; we will call such a path a \emph{level-set path}. (This Morse function should not be confused with the Morse function $f$ appearing in the mountain-pass lemma, whose role is played here by the area functional.) A path $X(t)$ homotopic to $\bar X(t)$ is called a \emph{sweep-out}. Any two Morse functions on $\tilde N$ define homotopic paths, so the definition of a sweep-out is independent of the choice of Morse function.

\subsection{Bulge surface: definition \& properties}
\label{sec:bulge}

The construction defining the bulge surface in \cite{Brown:2019rox} is similar to, but not exactly the same as, the above Almgren-Pitts construction. Specifically, the minimization is over level-set paths, rather than sweep-outs. Assuming, as we will, that the minimax exists and is an extremal surface, it must have index 1. We will call this surface $X^b$. For a generic metric, $X^b$ is unique; for a metric that is generic up to isometries, $X^b$ is unique up to the action of the isometries.

Any level-set path is a sweep-out, but the converse does not hold; specifically, whereas the level sets of a function cover each point of $\tilde N$ exactly once, a sweep-out may ``back up'' and cover some part of $\tilde N$ more than once. This leads to the question of whether the two prescriptions are equivalent. We believe they probably are equivalent, but do not have a proof. It is hard to see how the freedom to ``back up'' afforded by the sweep-outs could allow one to achieve a lower maximal area than the level-set paths; for this happen, the sweep-out would somehow have to slip through the bulge surface using only smaller surfaces, which intuitively seems impossible. However, we readily admit that this claim may simply reveal a lack of imagination on our part. If the two prescriptions are not equivalent, then one would have to show that the bulge prescription is actually well-defined and yields an extremal surface.

The rest of the paper is concerned with finding bulge surfaces in various holographic spacetimes. Of course, literally following the minimax definition of the bulge surface involves minimizing over an infinite-dimensional space of Morse functions, which is prohibitive, so instead one simply looks for index-1 surfaces. If there are multiple index-1 surfaces, however, how does one determine which one is the bulge? We will now prove a couple of lemmas that will help with this task.

First, let $X\in\tilde{H}$ be an index-1 surface, and recall that the index is additive under disjoint union. Therefore, if $X$ is disconnected, then each component must be extremal, and exactly one of them must have index 1, with the rest having index 0. The index-0 components may coincide with components of the surfaces $X_{0,1}$ that bound $\tilde N$. The index-1 component must lie in the interior of $\tilde N$.

For simplicity, from here on we will focus on connected surfaces; the disconnected case can be handled using the above decomposition.

\begin{lemma}\label{lem:intersect}
If connected index-1 surfaces $X_{2,3}\in\tilde{H}$ do not intersect, then there exists an index-0 surface $X_4\in\tilde{H}$, not equal to $X_0$ or $X_1$.
\end{lemma}
\begin{proof}
$X_2$ divides $\tilde N$ into two regions, one bounded by $X_0$ and $X_2$, the other bounded by $X_2$ and $X_1$. Since $X_3$ is connected, it lies entirely in one of these regions. Define $\tilde{\tilde N}\subset\tilde N$ as the region lying between $X_2$ and $X_3$. Let $\tilde{\tilde{H}}$ be the set of surfaces in $\tilde{\tilde N}$ homologous to $X_{2,3}$, and define $X_4$ as the least-area surface in $\tilde{\tilde{H}}$. $X_4$ cannot coincide with $X_2$ or $X_3$, since each of those surfaces has a negative mode that moves the surface into $\tilde{\tilde N}$. $X_4$ also cannot partially coincide with $X_2$ or $X_3$, since then its area could be reduced by rounding out the corner. Therefore $X_4$ must lie entirely in the interior of $\tilde{\tilde N}$, and must therefore be an index-0 surface.
\end{proof}

If $\tilde{H}$ does not contain any index-0 surfaces aside from $X_{0,1}$, then we say that $X_{0,1}$ are ``adjacent''.

\begin{lemma}\label{lem:leastarea}
If $X_{0,1}$ are adjacent, and all index-1 surfaces in $\tilde{H}$ are connected, then $X^b$ is the least-area one.
\end{lemma}
\begin{proof}
Let $X$ be an index-1 surface in $\tilde{H}$. Consider a restricted minimax problem where we minimize over level-set paths containing $X$, and for each path maximize the area. The solution is some index-1 surface that either equals $X$ or does not intersect $X$. The latter case is ruled out by lemma 1 and the assumption that $X_{0,1}$ are adjacent. Therefore $X$ maximizes the area on some level-set path.

$X^b$ is defined by minimizing, over level-set paths, the maximum on each path. We showed in the previous paragraph that every index-1 surface is a candidate in this minimization. Therefore $X^b$ is the one with the least area.
\end{proof}

We close this section with a explanation of the motivation for the minimax definition of the bulge surface in the setting of the python's lunch conjecture \eqref{eq:PLC}. The heuristic identification of $\tilde N$ with a tensor network requires us to view this partial Cauchy slice as a linear map $V:\mathcal{H}_0 \otimes \mathcal{H}_{\tilde{N}}\rightarrow \mathcal{H}_1$, between the Hilbert spaces associated to the cuts through the tensor network at the respective minimal surfaces $X_0$ and $X_1$, and additionally, the bulk Hilbert space on $\tilde{N}$, where we assume that $\text{Area}(X_0)<\text{Area}(X_1)$ without loss of generality. The above definition guarantees that the generalized entropy difference, $S_{\text{gen}}(X^b) - S_{\text{gen}}(X_1)$, controls the minimal amount of post-selection necessary to undo this map from $X_1$. To see this, one interprets the level sets for a given Morse function $\psi$ as providing a 1-parameter family of cuts of the putative tensor network. Different choices of function $\psi$ represent different cuts of the network, each of which includes a constrained maximum of the generalized entropy --- the area term in our case --- which we shall call $X^\psi$. Assuming that the local tensors in the network are generic enough,\footnote{For the states considered in Ref.\ \cite{Brown:2019rox} this amounts to having to wait until all of the perturbations are scrambled in the system.} each $\psi$ likewise represents a particular way of undoing the map $V$ layer-by-layer, by acting with few-body unitaries on each surface in the level set of $\psi$. The number of unitary operations needed to perform this contraction is parametrically controlled by the amount of post-selection present on the level set of $\psi$, which is given by $S_{\text{gen}}(X^\psi) - S_{\text{gen}}(X_1)$. Thus, the optimal way to undo the tensor network consists in minimizing over the choice of $\psi$, which leads to the motivation to define the bulge $X^b$ as the minimax surface between $X_0$ and $X_1$.

\section{Bulges break spatial isometries}
 \label{sec:breakisom}

 In this section we will explicitly show that the bulge can generally break the spatial isometries in spatially homogeneous states of the holographic system. We will do so for the cases where the bulges have spherical versus planar symmetry, and find that the latter indicates surprisingly low values of the complexity to reconstruct the lunch. We will also show that the bulge geometry, hence the complexity, is surprisingly sensitive to the choice of infrared regulator.

 \subsection{Spherical symmetry}
 \label{sec:spherical}

 	\begin{figure}[h]
 		\centering
 		\includegraphics[width = .4\textwidth]{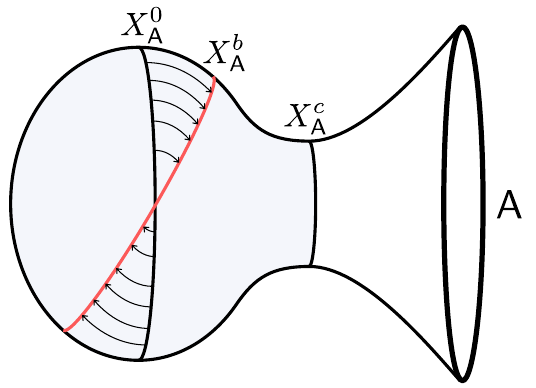}
 		\caption{Bulges, unlike minimal surfaces, can break the isometries of the Cauchy slice $\Sigma$. The naive bulge candidate $X^0_\R$ has additional negative modes, like the one represented by the arrows. There is an infinite family of bulges $X_\R^b$, which have Morse index $1$ and, additionally, zero modes corresponding to the isometries of $\Sigma$ connecting them. }
 		\label{fig:breakisom}
 	\end{figure}
 
 To be concrete, consider a holographic CFT placed on a spatial $S^{d-1}$ ($d>1$). Let $\ket{\Psi}$ be a homogeneous state on the sphere, with semiclassical description given by the initial data
 \be\label{eq:cauchyslicespherical}
 \text{d}s_\Sigma^2 = \text{d}\rho^2 + r^2(\rho) \text{d}\Omega_{d-1}^2\;,
 \ee
 specified at a moment of time reflection-symmetry $\Sigma$. We take the slice to have the topology of a $D$-dimensional ball, which implies $r\to0$ as $\rho\to0$; for smoothness, $r/\rho\to1$. Asymptotically, we demand that $r(\rho) \sim e^{\rho/\ell_{\rm AdS}}$ for $\rho \rightarrow \infty$, so that the spacetime is asymptotically AdS.  The RT surface for the full conformal boundary $\R$ is the empty set, $X_{\R}=\emptyset$. As illustrated in Fig.\ \ref{fig:breakisom}, in order to have a python, we require $r(\rho)$ to possess a local minimum at some $\rho_c>0$; this defines the constriction $X^c_{\R}$. Note that these surfaces both respect the spherical symmetry.
 
 From the considerations in section \ref{sec:minmax}, there must exist a bulge surface $X^b_\R$ between $X_\R$ and $X_\R^c$, in other words in the region $\rho<\rho_c$. Moreover, since $r(\rho)$ has two local minima, at $\rho=0$ and $\rho=\rho_c$, there must also exist a local maximum between them, at $0< \rho_{0} < \rho_c$, with
 \be
r_0:= r(\rho_0)\;,\qquad
 r'(\rho_0) =0\;,\qquad
 r''_0 :=r''(\rho_0)< 0\,.
 \ee
 We shall refer to this sphere as the ``naive bulge candidate'', denoted by $X^{0}_\R$.
 
The naive bulge candidate $X^{0}_\R$ is a totally geodesic surface, given that $r'(\rho_0) =0$, and thus it is extremal. The Jacobi operator \eqref{eq:Jacobi}, which determines its index, is easily computed:
 \be
 J = -\frac1{r_0^2}\bar\nabla^2-\alpha\;,
 \ee
 where $\bar\nabla^2$ is the Laplacian on the unit $(d-1)$-sphere and $\alpha$ is the following positive constant:
 \be
 \alpha = -(d-1)\frac{r''_0}{r_0}\,.
 \ee
 The eigenfunctions of $J$ are thus simply the spherical harmonics, and the eigenvalues are
 \be\label{eq:bulgeevalues}
\lambda_\ell=\frac{\ell(\ell+d-2)}{r_0^2}-\alpha\;,
 \ee
 where $\ell=0,1,\ldots$, with multiplicity 1 for $\ell=0$ and greater than $1$ for $\ell>0$.
 
 The $\ell=0$ eigenmode is always negative, $\lambda_0 =-\alpha< 0$, which simply corresponds to the uniform radial deformation, which decreases the surface's proper radius and therefore area. This means that $X^{0}_\R$ has index at least $1$.
 
 However, $X^{0}_\R$ may have index larger than $1$. The condition for the existence of additional negative modes is $\alpha >(d-1)/r_0^2$, or
 \be\label{eq:nmodephericalbulge}
-r''_0>\frac1{r_0}\,.
 \ee 
 The left-hand side in \eqref{eq:nmodephericalbulge} controls the intrinsic curvature of $X^{0}_\R$, while the right-hand side determines the curvature of  $\Sigma$ along the orthogonal direction.\footnote{
 	The condition \eqref{eq:nmodephericalbulge} is saturated if the metric on $\Sigma$ is that of a round $S^d$ in the neighborhood of $X^0_\R$, i.e.\ the profile function near $\rho_0$ is given by $r(\rho) = r_0 \sqrt{1-(\rho-\rho_0)^2/\rho_0^2}$. In this example $X^{0}_\R$ is the equatorial $S^{d-1}$, which has index $1$. In this borderline case, $X^{0}_\R$ has additional zero modes, the $\ell=1$ spherical harmonics, which represent the rotations of the $S^d$ broken by the equator.}  In particular, if the lunch is very prominent, then the latter will dominate and $X^{0}_\R$ will have index  larger than $1$. Since the true bulge must have index 1, in such a case, it must not be the naive one, $X^{b}_\R \neq X^{0}_{\R}$.\footnote{Note that, even when \eqref{eq:nmodephericalbulge} is false, so that $X^0_\R$ has index 1, it need not be the true bulge. Recall from subsection \ref{sec:bulge} that the bulge is the least-area index-1 surface in the relevant homology class. One can construct geometries in which $X^0_\R$ has index 1 but there is another index-1 surface with lower area, which is therefore the true bulge.}

What then is the true bulge, $X^{b}_{\R}$? It must be a surface that spontaneously breaks the $O(d)$ symmetry of $\Sigma$, as illustrated in Fig.\ \ref{fig:breakisom}. Furthermore, there must exist a family of bulge surfaces related by this symmetry, and the Jacobi operator on the bulge must have a zero-mode associated to the broken symmetry. 
 
 In $d=2$, the metric reduces to  $\text{d}s_\Sigma^2 = \text{d}\rho^2 + r^2(\rho)  \text{d}\varphi^2$, and the angular coordinate $\varphi$ has period $2\pi$. For a parametrization given by the parameter $\sigma$, the embedding function $ r = r(\sigma), \varphi =\varphi(\sigma)$  of a bulge can be found by extremizing the area functional
 \be 
 \text{Area} = \int \text{d}\sigma \sqrt{\dot{\rho}^2 + r^2 \dot{\varphi}^2}\;,
 \ee 
 where the dot represents $\text{d}/\text{d}\sigma$. Taking $\sigma$ to be the proper length of the bulge, the equation of motion for $\rho(\sigma)$ reduces to the equation of motion of a non-relativistic particle moving in one dimension with zero total energy
 \be\label{eq:eomgeodesics}
 \dot{\rho}^2 + V_{\text{eff}}(\rho) = 0\;,
 \ee 
 subject to the effective potential
 \be\label{eq:eomgeodesicsveff}
 V_{\text{eff}}(\rho) = \dfrac{r_m^2}{r^2(\rho)}-1\;.
 \ee 
 The parameter $r_m = r^2 {\dot \varphi}$ is a constant of motion along the trajectory. The particle starts at $r=r_m$, moves towards smaller values of the radius, and bounces back at $r(\rho) = r_m \leq r_0$. The parameter $r_m$ is fixed from the condition that the bulge $X^{b}_{\R}$ must be a closed trajectory, so that $\rho(\varphi)$ must be periodic. To select the minimal among all the possible bulge candidates, we impose that $X^{b}_{\R}$ wraps the circle once, so that $\rho(\varphi) = \rho(\varphi +2\pi)$. This imposes the constraint
 \be\label{eq:elapsedangle}
 \pi = \int^{\rho^R_m}_{\rho^L_m} \dfrac{r_m\text{d}\rho}{r^2\sqrt{-V_{\text{eff}}(\rho)}}\;,
 \ee 
 where $\rho_m^{L,R}$ are the roots of $r(\rho) = r_m$ closest to $\rho_0$ with $\rho_m^L < \rho_0 < \rho_m^R<\rho_c$. This constraint determines the value of $r_m$ as a function of the profile $r(\rho)$ of the geometry.
 
 That a solution to \eqref{eq:elapsedangle} exists, subject to \eqref{eq:nmodephericalbulge}, can be shown as follows. The frequency of small oscillations about $\rho_0$ is
 \be
 \omega = \sqrt{\frac{V_{\rm eff}''}{2}} = \sqrt{-\frac{r''_0}{r_0}}>\frac1{r_0}\,;
 \ee
since $\dot\phi=1/r_0$ in this limit, the distance in $\phi$ traversed over a half-period is less than $\pi$. On the other hand, as $\rho^R_m$ approaches $\rho_c$, the potential near the turning point flattens out, so the period goes to infinity. Between these two extremes, there therefore exists a value of $r_m$ obeying \eqref{eq:elapsedangle}. We can furthermore show that this solution has index 1. We first note that the Jacobi operator has a zero-mode, corresponding to the broken rotational symmetry; this mode has two nodes, at $\rho_m^{L,R}$, where the rotation acts tangent to the surface. Therefore this mode is the ``first excited state'' (in quantum mechanics parlance), so there is exactly one eigenmode with negative eigenvalue (the ``ground state''). By the same logic, a solution that oscillates $n$ times, if it exists, will have a zero-mode with $2n$ nodes and index $2n-1$. This guarantees that the solution we found is the only one with index 1, hence it must be the bulge.

 In higher dimensions, each bulge $X^{b}_{\R}$ will be a deformed sphere that spontanteously breaks the $O(d)$ symmetry of $\Sigma$ into some subgroup $H$. The bulges thus will have a number of zero modes given by $d(d-1)/2 - \text{dim}(H)$, according to Goldstone's theorem. What $H$ is ultimately will depend on the radial profile $r(\rho)$ of the python. However in the scenario with minimal but nonzero symmetry breaking, the negative modes of $X^{0}_{\R}$ will condense to preserve a $H = O(d-1)$ subgroup, with $d-1$ zero-modes. In this case the true bulge $X^b_\R$ is a squashed sphere along a particular axis, where the $d-1$ zero modes arise from the rotations of this axis.

 Similarly, the bulge can break discrete isometries of the holographic state, such as $\mathbf{Z}_2$ reflection symmetry, or permutation symmetry in the case of multiple boundaries. In section \ref{sec:bhint} we will show this latter case explicitly for microstates of three dimensional black holes.

 \begin{figure}[h]
    \centering
 	\includegraphics[width=0.35\textwidth]{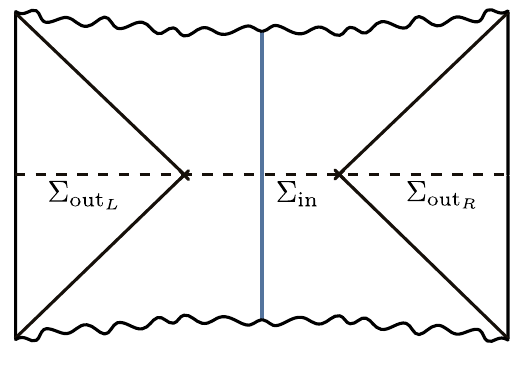}
    \hspace{1.5cm}
    \includegraphics[width=0.4\textwidth]{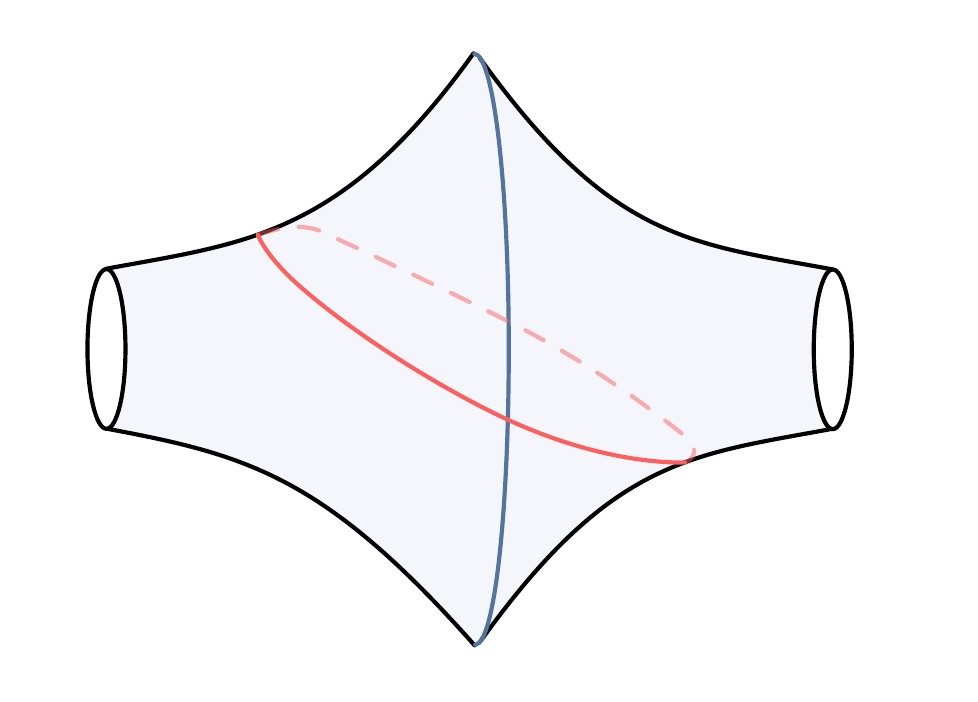}
	\caption{Specific example of a bulge which breaks spherical symmetry on a dust shell microstate of a two-sided black hole. On the left, the Penrose diagram of the geometry, where the trajectory of the dust shell is shown in blue. The semiclassical state is defined on the time reflection-symmetric Cauchy slice $\Sigma$, consisting of two exterior regions $\Sigma_{\text{out}_{L,R}}$ delimited by the apparent horizons of the two black holes, and a python's lunch geometry $\Sigma_{\text{in}}$ in the black hole interior. On the right, the geometry of the python $\Sigma_{\text{in}}$. The naive bulge is the maximum sphere, which sits at the position of the dust shell on $\Sigma$. The true bulge, in red, spontaneously breaks spherical and $\mathbf{Z}_2$ reflection symmetry of $\Sigma$.}
	\label{fig:dustshell}
\end{figure}

 We close with a specific example of a symmetry-breaking bulge, name that associated with semiclassical states of asymptotically AdS$_{d+1}$ black holes with interior dust shells, described in \cite{Balasubramanian:2022gmo,Balasubramanian:2022lnw}. In these geometries, there is a ``naive bulge'' (incorrectly identified in  \cite{Balasubramanian:2022gmo}\ as the true bulge) at the location of each spherical shell of dust, where the sphere carries the $SO(d)$ symmetry of the Cauchy slice. For the purposes of specificity and brevity, let us consider the example of a single shell in the interior of a two-sided black hole, as described in \cite{Balasubramanian:2022gmo}, in the case that the black hole on each side has the same temperature and the whole geometry has a $\mathbf{Z}_2$ reflection symmetry across the shell. The geometry is built by starting with two copies of the two-sided black hole, with the copies glued along the shell's trajectory as in Fig.\ \ref{fig:dustshell}, using the Israel junction conditions. The resulting background has time-reflection symmetry, and the metric along the time-symmetric slice $\Sigma$ corresponds to two copies of the metric on the $t = 0$ slice of the black hole, glued together at some proper radius $r_0$ in an exterior region of each. Let the radial coordinate in AdS-Schwarzschild coordinates be $y$ so that $y \to \infty$ is the boundary, the radius of the $S^{d-1}$ factor scales as $r(\rho) \sim e^{y/R_{AdS}}$, and the shell resides at $y = y_0$. If we choose a local radial coordinate $\rho$ which is zero at the shell, and for which the radial coordinate on each side of the shell is $y \sim y_0 - |\rho|$, then we have
 \begin{equation}
     r(\rho) = r_{\rm shell}  -r_0'|\rho|+ {\cal O}(\rho^2) \,,
 \end{equation}
where $r_0'>0$. As the metric of the python is of the form \eqref{eq:cauchyslicespherical}, $r''(\rho)$ has a delta-function singularity with negative coefficient at the shell location $\rho = 0$, and the condition in \eqref{eq:nmodephericalbulge} is automatically satisfied. Using the formula \ref{eq:bulgeevalues}, we can see that a surface coinciding with the shell has infinite index. A natural candidate bulge arises from considering each side of the shell to be the spatial slice of a cutoff AdS-Schwarschild geometry. One each side of the geometry, consder the minimal ``RT'' surface ending on an equatorial $S^{d-2}$ of the dust shell; and glue them together at the dust shell. There is a clear negative mode that arises from deforming the intersection of this surface with the dust shell off of the equator. This solution will retain an $SO(d-1)$ symmetry, and leave $d-1$ zero modes behind, corresponding to the choice of equator at which the bulge intersects the dust shell. 

\subsection{Planar symmetry \& simple interiors}
\label{sec:planar}

Planar-symmetric states can be obtained in the formal thermodynamic limit $r_0\rightarrow \infty$ from the spherical case. For the naive bulge candidate $X^{0}_\R$, whose topology is now $\mathbf{R}^{d-1}$, the spectrum of $J$ becomes continuous, and \eqref{eq:nmodephericalbulge} is never satisfied, provided that $r''(\rho_0)$ is kept finite in the scaling limit. In particular, this means that $X^{0}_\R$ for planar-symmetric states has index $\infty$, and is never the bulge surface.
 
If we regulate the transverse directions, unlike for minimal surfaces homologous to the full boundary, the character of the bulge will depend on the value of the IR cutoff. Consider for simplicity the $d=2$ case presented above, and decompactify the spatial circle, so that the geometry is 
\be 
\text{d}s_\Sigma^2 = \text{d}\rho^2 + r^2(\rho) \text{d}x^2\;,
\ee 
instead, for $x \in \mathbf{R}$. The function $r(\rho)$ is assumed to have a positive global minimum at $\rho=0$, a local maximum at $\rho_b>0$, and a local minimum at $\rho_c>\rho_b$; this defines a python, as above. (This is now a two-sided python, and the RT surface $X_\R$ at $\rho=0$ is no longer empty. This will not affect our calculation.) Assume that we regularize the geometry by adding a transverse IR cutoff at $x_\pm = \pm \Lambda_{\text{IR}}^{-1}/2$, and requiring the bulge surface to meet the IR cutoff surfaces orthogonally. This puts a Neumann boundary condition on the deformation function $\eta$ appearing in the second variation \eqref{eq:secondvariation}, ensuring that the Jacobi operator remains self-adjoint.

 For the naive bulge candidate $X^{0}_\R$, the Jacobi operator $J$ will have a spectrum given by $\lambda_s = (2\pi s \Lambda_{\text{IR}})^2  + r_0''\Lambda_{\text{IR}}$, for $s =0,1,2,...$. The eigenmodes correspond to the normal deformations $\eta_s = \cos(2\pi s\Lambda_{\text{IR}} x)$, which satisfy the boundary conditions set by the IR cutoff. Therefore the index of $X^{0}_\R$ will be greater than 1 if the IR cutoff is large enough,
 \be 
 \Lambda_{\text{IR}}^{-1} > - \dfrac{(2\pi)^2}{r_0''}\;,
 \ee 
 and the naive bulge will not be the correct one in these cases. 

  \begin{figure}[h]
 		\centering
 		\includegraphics[width = .93\textwidth]{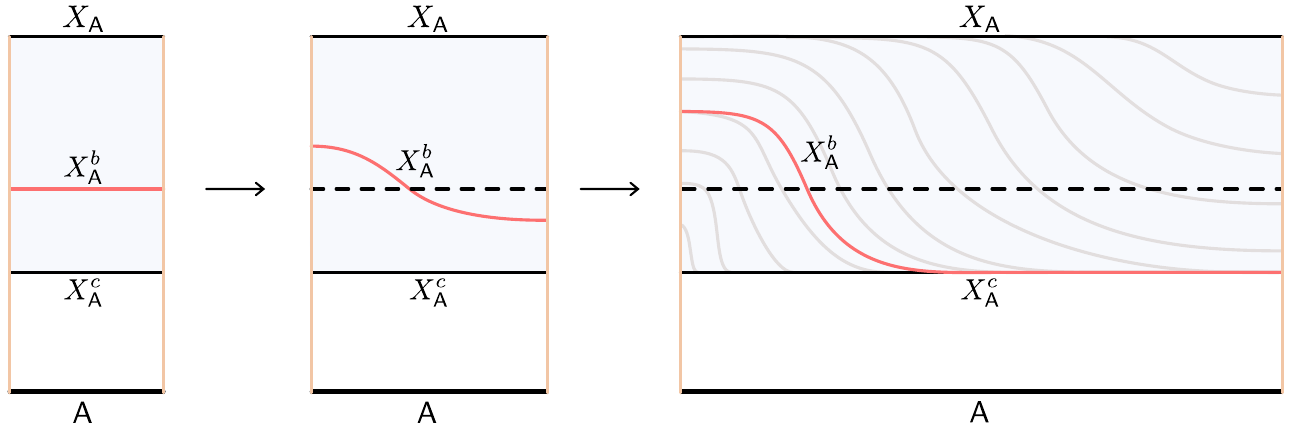}
 		\caption{As the IR cutoff $\Lambda_{\text{IR}}^{-1}$ is increased, the bulge undergoes a transition from the naive symmetric bulge to an extremal surface which spontaneously breaks $\mathbf{Z}_2$ reflection symmetry. In the thermodynamic limit, the bulge aproaches the constriction at any finite distance, and they only differ asymptotically in the transverse space direction. On the right figure, we include a minimax foliation of the lunch, in gray. For the sake of ilustration, we have omitted the additional space beyond $X_\R$ in case the python is two-sided. }
 		\label{fig:planarlunch}
 	\end{figure}

The true bulge will locally satisfy the equation of motion \eqref{eq:eomgeodesics} for the effective potential \eqref{eq:eomgeodesicsveff}, with the constant of motion $r_m = r^2 {\dot x}$. Accordingly, its endpoints will lie at the turning points of $V_{\text{eff}}(\rho)$, which lie at the same value of the radial coordinate $r(x_\pm) = r_m <r_0$. As $x$ goes from $x_-$ to $x_+$, $\rho$ will oscillate between the turning points, traversing that region $n$ times. From the properties of the minimax surface in section \ref{sec:minmax}, assuming that there is no other minimal surface homologous to $\R$, the true bulge will be the minimal index-1 surface in the lunch. The analysis of the index is most easily done by doubling the solution to obtain a solution to the periodic problem of the previous subsection, and retaining the negative modes of the latter solution that are invariant under a $\varphi\to2\pi-\varphi$ reflection. This solution will have $n$ oscillations, hence (as argued there) $2n-1$ negative modes, of which $n$ are reflection-invariant. So the only solution with index 1, and therefore the bulge, is the $n=1$ one. In particular, this means that the endpoints will lie on different sides of this surface, as shown in Fig.\ \ref{fig:planarlunch}.
 
 The endpoint radius $r_m$ is determined by the analog of \eqref{eq:elapsedangle} for the planar case, that is,
 \be\label{eq:ellapsedx}
 \Lambda_{\text{IR}}^{-1}= \int_{\rho^L_m}^{\rho^R_m} \dfrac{r_m\text{d}\rho}{r^2\sqrt{-V_{\text{eff}}(\rho)}}\;,
 \ee 
 where $\rho^{L,R}_m$ are the two solutions to $r(\rho) = r_m$ closest to $\rho_0$.  From \eqref{eq:ellapsedx}, it is easy to see that the area of the bulge will satisfy
 \be\label{eq:relation2dbulge} 
 \text{Area}(X^b_\R) - r_m \Lambda_{\text{IR}}^{-1}  =  \int_{\rho^L_m}^{\rho^R_m} \text{d}\rho \sqrt{-V_{\text{eff}}(\rho)}\;.
 \ee 
 
 In the thermodynamic limit $\Lambda_{\text{IR}}\rightarrow 0$, the endpoint asymptotes to the value of the radius at the constriction, $r_m \rightarrow r_c$. The right hand side of \eqref{eq:relation2dbulge} remains finite, since the integrand is everywhere finite in the corresponding domain of integration. The left-hand side precisely controls the exponent in the python's lunch conjecture \eqref{eq:PLC}. The complexity to reconstruct the black brane interior, according to the PLC \eqref{eq:PLC}, is
 \be\label{eq:complexityblackbrane} 
  \Co(\mathcal{R}) \sim  \exp \lambda_c\;.
 \ee 
 where 
 \be 
 \lambda_c = \dfrac{1}{8G}\int_{\rho^L_m}^{\rho^R_m} \text{d}\rho \sqrt{-V_{\text{eff}}(\rho)} \sim O(N^2\Lambda_{\text{IR}}^{0})
 \ee 
 provides a finite complexity density in the thermodynamic limit. 

 We arrive at the conclusion that the symmetry breaking of the bulge in the planar case makes the complexity to reconstruct the lunch \emph{not} scale exponentially with the coarse-grained entropy of the CFT system.\footnote{{Although the exponent does not scale with the entropy and system size, the subexponential volume factor in the complexity \emph{does} scale; it may be estimated as follows:
 \be 
 \mathcal{C}_V = \dfrac{\Lambda_{\text{IR}}^{-1}}{G\ell_{\rm AdS}} \int_0^{\rho_c} \text{d}\rho \,r(\rho) \sim N^2\Lambda_{\text{IR}}^{-1}\;.
 \ee 
}} The latter is determined by the generalized entropy of the constriction, $S_{\text{gen}}(X^c_\R) \sim O(N^2\Lambda_{\text{IR}}^{-1})$. In this naive sense, black brane interiors are ``simple'' to reconstruct. The direct tensor-network interpretation is that there is an optimal way to undo the tensor network from $\R$, which uses unitary operators that break planar symmetry, following an optimal foliation of the lunch which contains the bulge (see one such foliation in gray in Fig.\ \ref{fig:planarlunch}). This way, the amount of post-selection needed to undo the tensor network is drastically reduced compared to the planar-symmetric foliation of the lunch. 
 
 We note, moreover, that the notion of bulge is somewhat ambiguous when we remove the IR regulator of the transverse spatial directions. In the thermodynamic limit, the bulge sits on top of the constriction except close to one of its endpoints, at $x\rightarrow -\infty$ (for the $\mathbf{Z}_2$ reflection symmetric bulge the separation occurs at $x\rightarrow +\infty$). Furthermore, different IR regulators, e.g.\ those defined by moving the branes around by a spatial translation $x \rightarrow x+a$, will provide different bulges in the thermodynamic limit. All of them will hug the constriction at finite distance, since in the effective potential \eqref{eq:eomgeodesicsveff}, the turning point is near a maximum of the potential, which is where the particle is spending most of its time.\footnote{Had we decided to regularize the transverse space using periodic boundary conditions, we would have found a continuous family of bulges related by translational zero modes. As explained in section  \ref{sec:spherical}, in the $2+1$ dimensional case, these bulges cross the $r=r_0$ surface twice. Still, when $\Lambda_{\text{IR}}^{-1}\rightarrow 0$ all of the bulges practically hug the constriction except at a finite region, leading to a log-complexity, according to the python's lunch conjecture \eqref{eq:PLC}, which does not scale with the volume of the transverse space. This suggests that the effect of accumulation of the bulge on the constriction (and thus non-extensive log-complexity) is also present for large spherical boundaries, in the regime where the radial curvature of the lunch is much more prominent than the transverse curvature of the sphere.}

\section{Pythons in the vacuum}
\label{sec:vacuumpython}

We now move into the study of classical bulges that arise from the spatial entanglement structure of the ground state of the holographic CFT. Our main goal in this section is twofold. First, we provide an extensive study of the bulge for the entanglement wedge of two disks in AdS$_4$ in the connected phase. Second, we show that the bulge is significantly modified once a compact dimension AdS$\times Y$ with product metric is added. This effect leads to the resolution of singular bulges in AdS$_3$.

\vspace{-.5cm}

\subsection{$2+1$ dimensions}
\label{sec:vacuumads3}

\begin{figure}[h]
    \centering
    \includegraphics[scale =1]{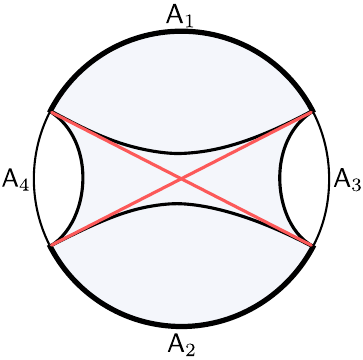}
    \caption{The entanglement wedge $\rew$ of two disjoint intervals $\R = \R_1 \cup \R_2$ in the connected phase contains a python. The bulge $X^{b}_{\R}$ is the red curve in the python, consisting of two intersecting geodesics.}
    \label{fig:ads3cross}
\end{figure}

The simplest example \cite{Brown:2019rox} of a classical python occurs in vacuum AdS$_3$. Consider the entanglement wedge $\rew$ of two intervals $\R = \R_1 \cup \R_2$ in the ground state of the CFT$_2$ on a spatial circle of radius $\ell$. If $|\R|\geq \pi \ell$, the two intervals comprise more than half of the boundary space and the RT surface $X_\R$ is in the connected phase (see Fig.\ \ref{fig:ads3cross}). The constriction consists of the disconnected minimal surface, $X^{c}_{\R} = X_{\R_1} \cup X_{\R_2}$. The bulge consists of two geodesics which cross each other, $X^{b}_{\R} = X_{\R_1 \R_3 } \cup X_{\R_3 \R_2}$, and is therefore singular. Since this surface is singular, the theory of small deformations reviewed in section \ref{sec:minmax} does not directly hold; in particular, neither the normal vector $n^\mu$ nor the extrinsic curvature $K$ is well defined at the crossing point, and deformations cannot in general be described in terms of a smooth function $\eta$.

Nonetheless, there is a sharp sense in which this surface has index 1. Its deformations can be divided into three classes: (1) those that leave a neighborhood of the intersection point unchanged; (2) those that move the intersection point, but leave the four segments connecting it to the boundary as geodesics; and (3) those that desingularize the intersection point. The first two kinds of deformations increase the total area of the surface, thus do not contribute to the index. The third kind includes two deformation directions, which make the surface homotopic to $X_\R$ and to $X^c_\R$ respectively. Both of these deformations decrease the area, and they should be thought of as the two directions that a negative mode can be turned on. Furthermore, since in both directions the area decreases already at first order in the deformation, the second derivative (formally, the eigenvalue of the Jacobi operator) is $-\infty$.

One way to think about this situation is as follows. The surfaces homologous to $\R$ fall into two homotopy classes, corresponding to the connected and disconnected phases. (There are other homotopy classes, but these are the only ones that are relevant for defining the bulge.) Each class contains a single extremum of the area, the RT and constriction respectively, which are local minima. These two homotopy classes meet along a codimension-one locus, consisting of surfaces with the cross topology. In terms of the area functional, this locus is a ridge, with the area having a finite negative slope moving away from the ridge in either direction. The minimal-area surface among the ones on the ridge is the cross consisting of two intersecting geodesics. There are no other extremal surfaces in this homology class.

The area of the bulge has a clear expression in terms of boundary entropies
\be 
\dfrac{\text{Area}(X^{b}_{\R})}{4G} = S(\R_1 \R_3) + S(\R_2 \R_3) = S(\R_1 \R_4) + S(\R_2 \R_4)\;,
\ee 
and thus the  
complexity according to \eqref{eq:PLC}, is 
\be\label{eq:rescads}
\mathcal{C}(\mathcal{R}) \sim \;\exp(\frac{\Icmutual}{2})\;,
\ee  
where $\Icmutual \equiv S(\R_1 \R_3) + S(\R_1  \R_4) - S(\R_1) - S(\R_1  \R_3  \R_4)$ is the the conditional mutual information. The conditional mutual information is non-negative by virtue of strong subadditivity of the von Neumann entropy. In the holographic context, this was proven in \cite{Headrick:2007km}.\footnote{{The geometric volume of the python $N$ is independent of the sizes of $\R_i$, by virtue of the the Gauss-Bonnet theorem. Namely, the python $N$ is bounded by four geodesics which meet perpendicularly to the asymptotic boundary, and thus the internal angles between them vanish, giving $\int_N \kappa = -2\pi$, where $\kappa$ is the Gaussian curvature.}}

\subsubsection{Excising 
boundary points}

As an additional observation, we note that the appearance of a python in the vacuum can even be sensitive to losing access to a measure-zero subset of the boundary Cauchy slice.

\begin{figure}[h]
\centering
 	\includegraphics[width=0.3\textwidth]{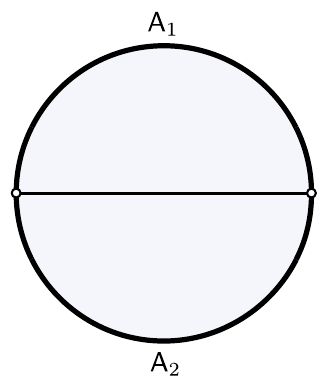}
    \hspace{1.5cm}
    \includegraphics[width=0.39\textwidth]{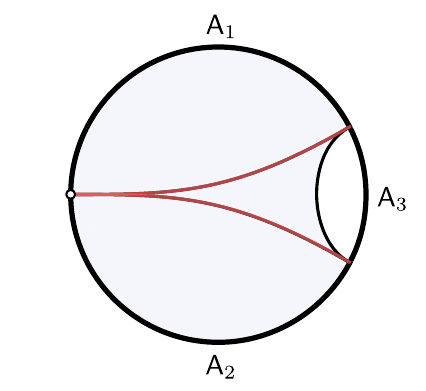}
	\caption{On the left, we excise two points of the full boundary. One can get this configuration as a limit from figure \ref{fig:ads3cross} where $\R_3 $ and $\R_4 $ have each been reduced to a point. This configuration can be recognized as vacuum AdS$_3$ in Rindler coordinates. On the right, removing a point from an interval generates a constriction and a python for the proper subregion $\R = \R_1 \cup \R_2$. One can get this configuration as a limit from figure \ref{fig:ads3cross} where $\R_4 $ has been reduced to a point. The bulge hugs the constriction and the difference in areas vanishes, giving no exponential complexity to reconstruct the lunch.}
	\label{fig:ads3_excised_points}
\end{figure}

We first consider excising two points of the asymptotic boundary of AdS$_3$ as illustrated in the left of Fig.\ \ref{fig:ads3_excised_points}. We can think of removing such points as beginning with a cutoff entangling surface and taking the cutoff to infinity. The situation is the same as describing AdS in Rindler coordinates. The acceleration horizon leaves a measure-zero causal shadow on the constant global time slice, granted that we can access to the rest of the boundary $\R_1 \cup \R_2$. There is no python in this case. As a check on this, we note that this is what we would get from the two-interval case if we had shrunk $\R_3,\R_4$ to antipodal points.

If, instead, we shrink just $\R_4$ to a point, leaving $\R_3$ as a finite interval, we obtain the situation shown on the right side of Fig.\ \ref{fig:ads3_excised_points}. The region $\R_1\cup\R_2$ can be thought of as an interval with a point removed from its interior. Removing this point has the effect of creating a python. However, the bulge coincides with the constriction, $X^b_\R=X_\R^c=X_{\R_1}\cup X_{\R_2}$, implying a vanishing exponent in the complexity (leaving only the subexponential volume factor).

\begin{figure}[h]
\centering
    \includegraphics[width=0.34\textwidth]{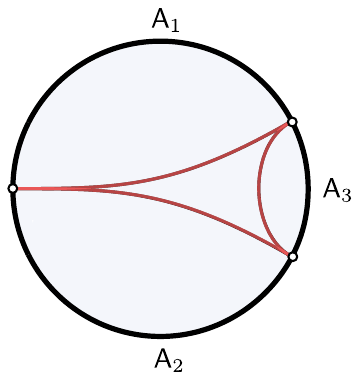}
    \hspace{1.5cm}
    \includegraphics[width=0.37\textwidth]{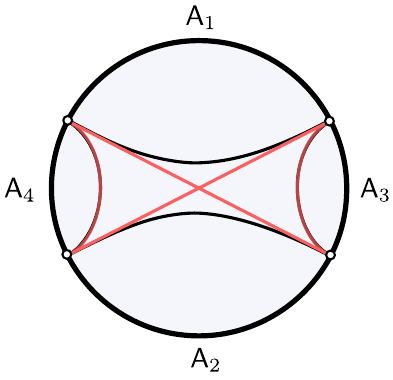}
	\caption{Excising three or more points generates a python for the full boundary. On the left, the case of three points, where the bulge is the union of the RT surfaces for $\R_1$, $\R_2$, and $\R_3$. The bulge coincides with the constriction, 
 giving no exponential complexity to reconstruct the python. On the right, excising four points, the bulge is the bulge for $\R_1\cup \R_2$ plus the RT surface for $\R_1\cup \R_2$.}
	\label{fig:ads3_4_excised_points}
\end{figure}

We could also include $\R_3$ in the region, but still leave out its endpoints, so it now covers the entire boundary save 3 points, as shown in the left part of Fig.\ \ref{fig:ads3_4_excised_points}. Again, there is a python, and again the bulge coincides with the constriction, $X^b_\R=X_\R^c=X_{\R_1}\cup X_{\R_2}\cup X_{\R_3}$, leaving a vanishing exponent in the complexity.

Finally, with four points excised, the situation is slightly more complicated, given that the bulge for $\R$ will only partially coincide with the constriction, as shown in Fig.\ \ref{fig:ads3_4_excised_points}. Assuming that $\R_1 \cup \R_2$ is large enough to have a connected entanglement wedge, the bulge for $\R$ will be the cross for $\R_1 \cup \R_2$ plus the RT for $\R_1\cup\R_2$ (which is also the RT for $\R_3\cup\R_4$), while the constriction is the union of the RTs of the individual intervals:
\be
X_\R^b = X_{\R_1\cup \R_3}\cup X_{\R_1\cup\R_4}\cup X_{\R_3}\cup X_{\R_4}\,,\qquad
X_\R^c = X_{\R_1}\cup X_{\R_2}\cup X_{\R_3}\cup X_{\R_4}\,.
\ee
This points to a very interesting property of the PLC: the log complexity to reconstruct the python with access to $\R$ is the same as the log complexity to reconstruct python with only $\R_1 \cup \R_2$. We will come back to this property when studying the black hole interior in section \ref{sec:bhint}.

Thus, according to the PLC, in the vacuum of a 2d CFT, removing 4 or more points from the boundary produces an exponential complexity, while removing fewer than 4 points does not. This very specific prediction would seem to provide a useful target for testing the PLC, if the complexity can somehow be independently estimated.

\subsection{Higher dimensions}

We will now move to the higher dimensional pythons in vacuum AdS$_{d+1}$ for $d>2$, for a boundary subregion consisting of two disks $\R = \R_1 \cup \R_2 \subset S^{d-1}$. Consider that the sizes of the disks and their positions are such that the entanglement wedge corresponding to this subregion is connected, as in the previous case. The constriction consists again of the disconnected surface, $X^{c}_{\R} = X_{\R_1} \cup X_{\R_2}$. Our case study will be AdS$_4$, where we will explicitly find the bulge $X^{b}_{\R}$. As we will show, unlike in AdS$_3$, $X^{b}_{\R}$ is a smooth surface, a property that is expected to hold for any dimension $d>2$.

\subsubsection{Warmup: flat $\mathbf{R}^3$}

\begin{figure}[h]
    \hspace{.3cm}
    \includegraphics[scale =.9]{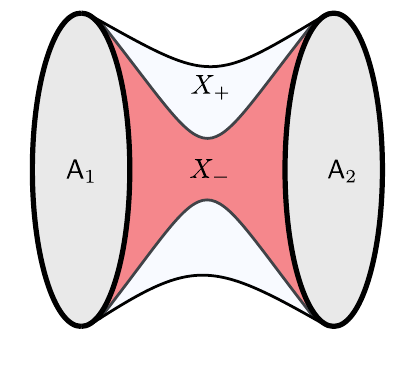}
    \hspace{1.2cm}
    \includegraphics[width= .5\textwidth]{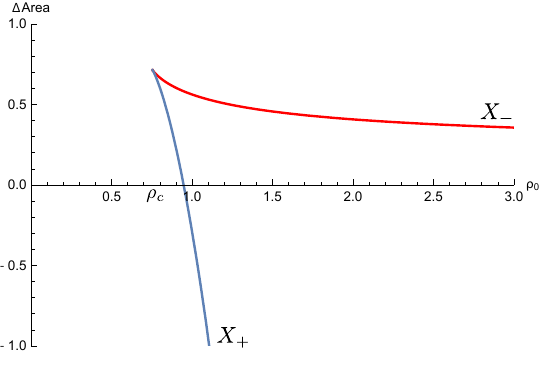}
    \caption{On the left, the three extremal surfaces anchored to the two circles $\partial \R$ consist of the disks themselves, $\R$, and two catenoids $X_\pm$, where $X_-$ is represented in red. On the right, the phase diagram of $\Delta \text{Area} (X_\pm) = \text{Area} (X_\pm) - \text{Area}(\R)$ for the two catenoids, as a function of the disk radius $\rho_0$. For $\rho_0\gg1$, this is a simplified model of a python, where the RT is $X_{\R} = X_+$, the constriction is $X_{\R}^{c}=\R$, and the bulge is $X_{\R}^{b}= X_-$.}
    \label{fig:catenoidphase}
\end{figure}

The features of the bulge surface $X^{b}_{\R}$ are qualitatively captured by the simpler model of two parallel and coaxial disks $\R = \R_1 \cup \R_2$, both of radius $\rho_0$, in $\mathbf{R}^3$. The disks are separated by a distance $z_0=1$ which sets the scale of the system. It is convenient to use adapted cylindrical coordinates,
\be
\text{d}s_\Sigma^2 =\text{d}z^2+\text{d}\rho^2+\rho^2\text{d}\varphi^2\;,
\ee
where the disks lie at $z = \pm\frac{1}{2}$, respectively. For sufficiently large radius, $\rho_0>\rho_c\approx 0.754$, there are three extremal surfaces homologous to $\R$: the disks themselves, and two catenoids, which we call $X_\pm$:
\be
\rho=a_\pm\cosh\frac z{a_\pm}\;,
\ee
where the parameters $a_\pm$ are the larger and smaller solutions to the boundary condition
\be\label{catbc}
\rho_0 = a\cosh\frac1{2a}\,.
\ee
These obey $a_- < a_c < a_+$, where $a_c\approx0.417$ is the solution to $\tanh(1/(2a))=2a$. For $\rho_0<\rho_c$, there are no solutions to \eqref{catbc}, and the disks are the only extremal surface, while for $\rho_0=\rho_c$ there is one solution. These solutions are illustrated in Fig.\ \ref{fig:catenoidphase}. In the limit in which the disks are very large $\rho_0\gg 1$, the catenoid $X_+$ becomes approximately cylindrical, $a_+ \approx \rho_0$, while $X_-$ pinches off, $a_- \sim (\log \rho_0)^{-1}$.

The area of $X_\pm$ is given by
\be
\area(X_\pm) = \pi a_\pm \left(1 + a_\pm \sinh \left(a^{-1}_\pm\right)\right).
\ee 
The phase diagram is represented in Fig.\ \ref{fig:catenoidphase}. $X_-$ always has larger area that the other two surfaces, while $X_+$ and $\R$ switch, with $X_+$ larger for $\rho_0$ less than about 0.948.

Now let us consider the indices of these surfaces. Intuitively, it is clear that the flat surface $\R$ is stable; this also follows from the fact that its Jacobi operator \eqref{eq:Jacobi} is simply minus the Laplacian, which clearly has no negative modes on the disk with Dirichlet boundary conditions. Consider starting with $\rho_0<\rho_c$ and continuously increasing $\rho_0$. At $\rho=\rho_c$, a new critical point of the area functional appears and bifurcates. From a Morse theory perspective, we expect that the smaller of the two, $X_+$, should have index 0 (hence the analogue of the constriction) and the larger one, $X_-$, index 1 (the analogue of the bulge). This is supported by the fact that the larger-area one sits geometrically between the two smaller-area ones. We will now show that it is indeed the case that $X_+$ has index 0 and $X_-$ has index 1.

The Jacobi operator \eqref{eq:Jacobi} takes the following form on the catenoid, with $a$ being either $a_+$ or $a_-$:
\be
J = -\frac{\sech^2(z/a)}{a^2}\left(a^2\partial_z^2+\partial_\varphi^2+2\sech^2(z/a)\right)
\ee
Defining $\zeta=z/a$ (which ranges from $-1/(2a)$ to $1/(2a)$) and fixing a mode in the $\varphi$ direction, $\eta(\zeta,\varphi)=e^{in\varphi}f(\zeta)$, the eigenvalue equation becomes
\be
\label{eq:Jnf}
J_n f:=\sech^2\zeta\left(-f''+\left(n^2-2\sech^2\zeta\right)f\right)=a^2\lambda f\,.
\ee
The parameter $a$ now enters only in setting the boundary condition and rescaling the eigenvalue. The operator $J_n$ is self-adjoint with respect to the inner product
\be
\ev{f,g}_a:=\int_{-1/(2a)}^{1/(2a)}{\text{d}}\zeta\,\cosh^2\zeta \,f^*g
\ee
($d\zeta\,\cosh^2\zeta$ being the area element on the catenoid), so it admits a negative eigenvalue if and only if there exists a function $f$ such that $\ev{f,J_n f}_a<0$. This implies that, if a negative eigenmode $f_1$ exists for some $a=a_1$, then a negative eigenmode must also exist for any $a_2\le a_1$, since the function $f_2$ which is equal to $f_1$ on the interval $[-1/(2a_1),1/(2a_1)]$ and 0 outside of it, obeys $\ev{f_2,J_n f_2}_{a_2}=\ev{f_1,J_n f_1}_{a_1}<0$. (Note that $f_2$ is continuous by virtue of the Dirichlet boundary condition on $f_1$ at $\pm1/(2a_1)$. Note also that $f_2$ is not itself an eigenmode, and this argument does not tell us the value of the negative eigenvalue, only its existence.) This argument holds separately for each $\zeta$-parity sector.

The quantity $\ev{f,J_n f}_a$ is the same as the energy expectation value (with respect to the usual $L^2$ norm) for a particle on the interval $[-1/(2a),1/(2a)]$ with wave function $f$, subject to the potential
\be
V_n(\zeta) = n^2-2\sech^2\zeta\,.
\ee
So $J_n$ has a negative mode if and only if the particle has at least one negative-energy state. This potential is solvable in the $a\to0$ limit; the ground state wave function is $\sech\zeta$, with energy $n^2-1$. For $n\ge1$, this is non-negative, ruling out negative-energy states for any $a$. For $n=0$, the ground state is the only bound state. By continuity, this suggests the existence of a negative mode for small $a$, which we will confirm below. On the other hand, the lack of an odd negative-energy state rules out an odd negative mode at any $a$, limiting the index to at most 1 (since the first excited state is odd).

Next, we note that for $a=a_c$, $J_0$ has a zero mode: $f_c=1-\zeta\tanh\zeta$. This rules out negative modes for $a>a_c$, showing that $X_+$ has index 0. For $a<a_c$, $f_c$ can be smoothed out near $\zeta=\pm1/(2a_c)$, decreasing $\ev{f_c,J_0f_c}_a$, proving the existence of a negative mode, and confirming that $X_-$ has index 1.

If one is mainly interested in the limit $a\to0$, it turns out that there is a mathematically elegant way to prove that $X_-$ has index $1$. The trick is to conformally map $X_-$ to the Riemann sphere using the \emph{Gauss map} $n: X_- \rightarrow S^2 \setminus \{\text{N,S}\}$, where $\text{N,S}$ denote the north and south pole, respectively. We leave the details of this map for appendix \ref{app:B}. The central property which makes things simpler is that under this map, \eqref{eq:secondvariation} acquires the following form on the sphere
\be\label{eq:gaussvariationcatenoid}
\delta^{(2)}\text{Area}(S^2\setminus \{\text{N,S}\}, \tilde{h}) = -\dfrac{1}{2}\int_{S^2\setminus \{\text{N,S}\}} {\text{d}\Omega}\,\sqrt{\tilde{h}}\,\eta \left(\tilde{\nabla}^2 + 2 \right)\eta \;,
\ee 
where $\tilde{h}_{\mu\nu} = (\sin^2 \theta)\, h_{\mu\nu}$ is the round metric on $S^2 \setminus \{\text{N,S}\}$, for the polar angle $\theta$. As noted in appendix \ref{app:B}, we can extend the metric $\tilde{h}_{ij}$ smoothly to the two poles {N,S} and hence, the index($X_-$) = index($S^2$) with the metric $\tilde{h}_{ij}$ and the quadratic form \eqref{eq:gaussvariationcatenoid}. It is straightforward to evaluate the index of the second order differential operator \eqref{eq:gaussvariationcatenoid} in $L_{\tilde{h}}^2(S^2)$. The eigenvectors correspond to spherical harmonics $(-\tilde{\nabla}^2 - 2 )Y_{\ell,n} = \tlambda_{\ell,n} Y_{\ell,n}$\;, with eigenvalues $\tlambda_{\ell,n} =\ell(\ell+1)-2$. Therefore, there is a single negative mode $\lbrace Y_{0,0} \rbrace $, and three zero modes $\lbrace Y_{1,0},Y_{1,\pm 1}\rbrace$, making the differential operator \eqref{eq:gaussvariationcatenoid} of index $1$ in $L_{\tilde{h}}^2(S^2)$.

A numerical analysis reveals that the negative eigenvalue of $J_0$ on the line is $\lambda_0\approx -0.564$. The eigenfunction falls off quickly at large $\zeta$, so the eigenvalue is essentially unchanged at small $a$. According to \eqref{eq:Jnf}, the eigenvalue of $J$ is then $\lambda\approx\lambda_0/a^2$.

\subsubsection{Two disks in AdS$_4$}

We can now extend the previous analysis to full-fledged AdS$_4$, where the bulk Cauchy slice is $\mathbf{H}^3$, for two boundary disks $\R = \R_1 \cup \R_2 \subset S^2$ with an entanglement wedge $\rew$ in the connected phase. Following \cite{Krtous:2014pva} (see also \cite{Fonda:2015nma}) we begin in adapted cylindrical coordinates in $\mathbf{H}^3$ (and set $\ell_{\text{AdS}}=1$)
\be \label{eq:h3metric}
\text{d}s^2 = \frac{1}{1+P^2}\,\text{d}P^2 + (1+P^2)\,\text{d}z^2 +  P^2\text{d}\varphi^2\;,
\ee
where constant-$z$ slices now corresponds to hyperbolic disks $\mathbf{H}^2$, with $P$ a radial coordinate on each disk. In the conformal frame adapted to these coordinates, the spatial boundary is just the flat cylinder $ \mathbf{R} \times S^1$, and the subregion $\R$ corresponds to two semi-infinite cylinders, $\R =\lbrace z\leq z_1\rbrace \cup \lbrace z\geq z_2\rbrace$. We shall consider the reflection-symmetric case $z_1 = -z_2 =: z_0$ without loss of generality. 

In this configuration, the constriction $X_{\R}^{c}$ will correspond to the disconnected surface, consisting of the two hyperbolic disks, $X_{\R}^{c} = \lbrace z=z_0 \rbrace \cup \lbrace z=-z_0 \rbrace$. The rest of the extremal surfaces correspond to two connected ``catenoids'' $X_\pm$ with embedding functions $(P, z_\pm (P), \varphi )\in \mathbf{H}^3$ given by (cf. \cite{Krtous:2014pva}) 
\be \label{eq:cath3}
z_\pm(P) =  \alpha_\pm  F\left(\arccos \frac{ a_\pm}{P}, \frac{1+a^2_\pm}{1+ 2 a^2_\pm}\right) - \beta_\pm \Pi\left(\frac{1}{1+a^2_\pm},\arccos \frac{a_\pm}{P} , \frac{1+a^2_\pm}{1+2a^2_\pm}\right)
\ee 
where $F(\phi, m)  = \int_0^\phi (1-m\sin^2 \theta)^{-1/2}d\theta,$ and $\Pi(n,\phi,m) = \int_0^\phi (1-n\sin^2\theta)^{-1}(1-m\sin^2 \theta)^{-1/2} d\theta$ are elliptic integrals of the first and third kind, respectively, and $\alpha_\pm = \frac{a_\pm (1+a^2_\pm)}{\sqrt{(1+a^2_\pm)(1+2 a^2_\pm)}}$ and $\beta_\pm = \frac{a^3_\pm}{\sqrt{(1+a^2_\pm)(1+2 a^2_\pm)}}$ are two constants. Imposing the boundary condition $z(\infty) = z_0$, we can numerically solve for the two values of the throat size $a_\pm$ as a function of the separation of the two disks. The connected solutions $X_\pm$ exist as long as $z_0 \leq z_c \approx 0.501$. They satisfy $a_- < a_c < a_+$ for $a_c \approx 0.538$.

\begin{figure}[h]
    \hspace{.3cm}
    \includegraphics[scale =.87]{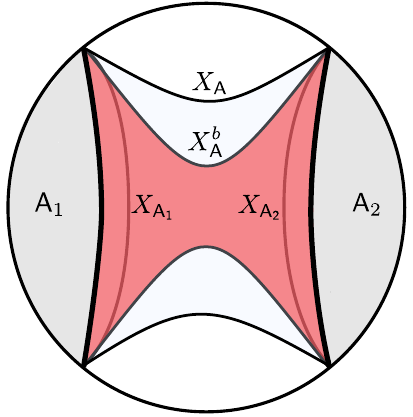}
    \hspace{1cm}
    \includegraphics[width= .52\textwidth]{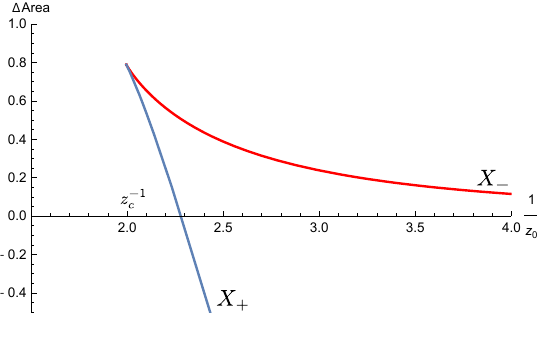}
    \caption{On the left, the three extremal surfaces anchored to the two circles $\partial R$ consist of the disconnected minimal surface $X_{\R}^{c} = X_{\R_1}\cup X_{\R_2}$, the RT surface $X_{\R} = X_+$ and the bulge $X_{\R}^{b} = X_-$. On the right, the phase diagram of $\Delta \text{Area} (X_\pm) = \text{Area} (X_\pm) - \text{Area} (X_{\R}^{c})$ as a function of the boundary separation $z_0$.
    }
    \label{fig:catenoidphaseads}
\end{figure}

The area difference $\Delta \text{Area}(X_\pm) = \text{Area}(X_\pm)- \text{Area}(X_{\R}^{c})$ is finite, and given by 
\be 
\Delta \area(X_\pm) = 4\pi \left[1 + \dfrac{a^2_\pm }{\sqrt{1+2a^2_\pm}}K\left( \frac{1+a^2_\pm}{1+2a^2_\pm}\right) - \sqrt{1+ 2a^2_\pm}E\left( \frac{1 + a^2_\pm}{1+2 a^2_\pm}\right) \right],
\ee 
where $K, E$ are complete elliptic integrals of the first and second kind, respectively. As we show in Fig.\ \ref{fig:catenoidphaseads}, the RT surface $X_\R$ for $z_0$ less than about $0.438$ corresponds to $X_+$.

As in the flat space example of the previous subsubsection, $X_+$ has index 0 while $X_-$ has index 1; since this is the only index-1 surface, it must be the bulge. To show this, we will employ the same technique in AdS$_4$ as we did for finding the negative mode of the catenoid in $\mathbf{R}^3$. 

Evaluating the Jacobi operator in the coordinate system \eqref{eq:h3metric} for which the embedding function of the catenoid is \eqref{eq:cath3}, leads to the eigenvalue equation 
\be
\frac{1}{P^2}\left((P^4+P^2-a^2-a^4)\frac{\partial^2}{\partial P^2} + P(1+2P^2)\frac{\partial}{\partial P} + \frac{\partial^2}{\partial \varphi^2} +\frac{2a^2(1+a^2)}{P^2}  - 2P^2\right)\eta = -\lambda \eta \, ,
\ee
where we have set $a_\pm$ in \eqref{eq:cath3} to $a$ for brevity. Note that in this coordinate system, $P \in [a,\infty)$. Using axial symmetry, we can assume the eigenfunctions to be of the form $\eta_n(P,\varphi) = e^{in\varphi}f_n(P)$. We now perform the change of variables and rescaling
\begin{gather}
\rho = \frac{1}{2}\log \left( \frac{1+2P^2 + 2 \sqrt{P^4+P^2-a^2-a^4}}{1 + 2a^2} \right) \;, \\
f_n(P) = \left(\frac{2a^2}{(1+2a^2)\cosh{2 \rho}-1}\right)^{1/4}g_n(\rho)\;,
\end{gather} 
to get a one-dimensional Schrödinger problem,
\be\label{eq:schrodingerAdS}
\left[-\frac{\text{d}^2}{\text{d}\rho^2} + V_{\text{eff}}(n,\rho)\right]g_n(\rho) = \lambda g_n(\rho)\;,
\ee
in terms of the effective potential 
\be
V_{\text{eff}}(n,\rho) = \frac{9}{4} - \frac{5(a^2+a^4)}{((1+2a^2)\cosh{2 \rho} -1)^2} - \frac{1-4n^2}{2((1+2a^2)\cosh{2\rho}-1)}\;.
\ee

\begin{figure}[h]
    \centering
    \includegraphics[width= 0.5\textwidth]{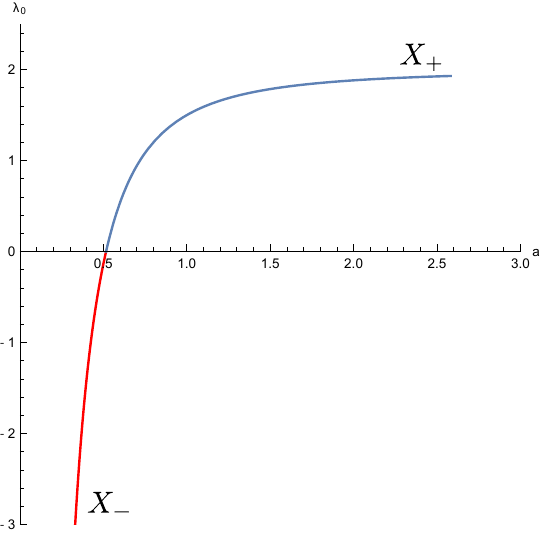}
    \caption{Smallest eigenvalue $\lambda_0$ of the Jacobi operator $J$ as a function of throat size $a$. The extremal surface $X_-$ with $a < a_c \approx 0.538$ always has one negative eigenvalue and therefore has index 1. The extremal surface $X_+$ with $a>a_c$ has no negative eigenvalues and therefore has index 0, i.e. it is locally minimal.}
    \label{fig:negavtiveeigen}
\end{figure}

We can now solve for the smallest eigenvalue $\lambda_0$ as a function of $a$ numerically by a shooting method. For $n = 0$, there is a negative eigenvalue if $a < a_c \approx 0.538$, as illustrated in Fig.\ \ref{fig:negavtiveeigen}. This precisely corresponds to $X_-$, and shows that it has Morse index larger than $0$. For $n>0$ there are no negative eigenvalues. Therefore, $X_-$ has index $1$ while $X_+$ has index $0$.

\subsection{Compact dimensions} 

Top-down string theory constructions of duals of CFTs take the form $\text{AdS}_d \times Y$, where $Y$ typically contains a factor whose size is of order the $\text{AdS}$ scale. For example, the type IIB string dual to the D1-D5 system takes this form with $d = 3$ and $Y =  S^3\times T^4$ (or $  S^3\times$K3), where the radius of the $ S^3$ is the $\text{AdS}$ scale and the $T^4$ or K3 are of order the string scale.

We would like to investigate how the extra dimensions affect the RT, constriction, and bulge surfaces. For a product metric $\text{d}s^2 = \text{d}s_M^2 + \text{d}s^2_Y$ on a product manifold $M\times Y$, an extremal surface $X$ in $M$ lifts to an extremal surface $X\times Y$, since the components of the extrinsic curvature in the $Y$ directions vanish. However, since the Jacobi operator will have KK modes on $Y$, the index of $X$ and $X\times Y$ may be different. More precisely, the KK modes contribute positively to the Jacobi operator; hence if the index of $X$ is 0 then the index of $X\times Y$ is also 0. Therefore, the RT and constriction computed on $M$, and lifted to $M\times Y$, are candidates for the RT and constriction on the full space. For the RT, we can prove that this is correct surface.

\begin{lemma}
Let $X_\R$ be the minimal-area surface in $M$ homologous to the boundary region $\R$. Then $X_\R\times Y$ is the minimal-area surface in $M\times Y$ homologous to $\R\times Y$.
\end{lemma}
\begin{proof}
Let $\tilde X$ be a surface in $M\times Y$ homologous to $\R\times Y$. For each point $y\in Y$, define $\tilde X(y)\in M$ as the intersection of $\tilde X$ with $M\times\{y\}$, $\tilde X(y):=\{x\in M:(x,y)\in\tilde X\}$. $\tilde X(y)$ is homologous to $\R$ (via the intersection of the homology region for $\tilde X$ with $M\times\{y\}$). Therefore $\area(\tilde X(y))\ge\area(X_\R)$, so we have
\be
\area(\tilde X)\ge\int_Y{\text{d}}y\sqrt{g_Y}\area(\tilde X(y))\ge
\text{Vol}(Y)\area(X_\R)=\area(X_R\times Y)\,.
\ee
\end{proof}

Via maximin \cite{Wall:2012uf}, this statement extends to the HRT formula. It therefore justifies the standard practice in the holographic entanglement literature of ignoring the extra dimensions. While we don't have a proof, we suspect that the same holds for the constriction. In fact, we conjecture that the \emph{only} index-0 surfaces in $M\times Y$ are those of the form $X\times Y$ for some index-0 surface $X$ in $M$.

The situation is very different at index 1. If $X$ is an index-1 surface in $M$, then $X\times Y$ has index at least 1; but the KK modes may lead to several negative modes, hence an index larger than 1. One therefore has to search for the bulge among surfaces that are not products, but that genuinely probe the extra dimensions. 
Whether this happens depends on the ratio between the negative eigenvalue of $X$ and the KK scale, which in turn is determined by the size of $Y$. Usually, the negative eigenvalue of $X$ is roughly determined by the AdS radius; so if $Y$ is much smaller than the AdS scale then $X\times Y$ has index 1, but if it is AdS-sized (as in the case of AdS$_3\times  S^3$ and AdS$_5\times  S^5$), its index may not be 1. However, our very first example, treated in subsubsection \ref{sec:flatS1} will be an exception, as the negative eigenvalue of $X$ in that case is $-\infty$, leading to a non-trivial bulge for extra dimensions of any size.

\subsubsection{Warmup: $\mathbf{R}^2\times  S^1$}
\label{sec:flatS1}

As we discussed in subsection \ref{sec:vacuumads3}, the bulge for two boundary intervals in AdS$_3$ is a cross, and its negative eigenvalue is $-\infty$. Therefore, for AdS$_3\times  S^1$, the cross times $ S^1$ is an extremal surface with infinite index, hence cannot be the bulge. To simplify this situation, we approximate it with $\mathbf{R}^2\times  S^1$ --- appropriate for example if the $ S^1$ is much smaller than the AdS$_3$.

\begin{figure}[h]
    \centering
    \includegraphics[width = .3\textwidth]{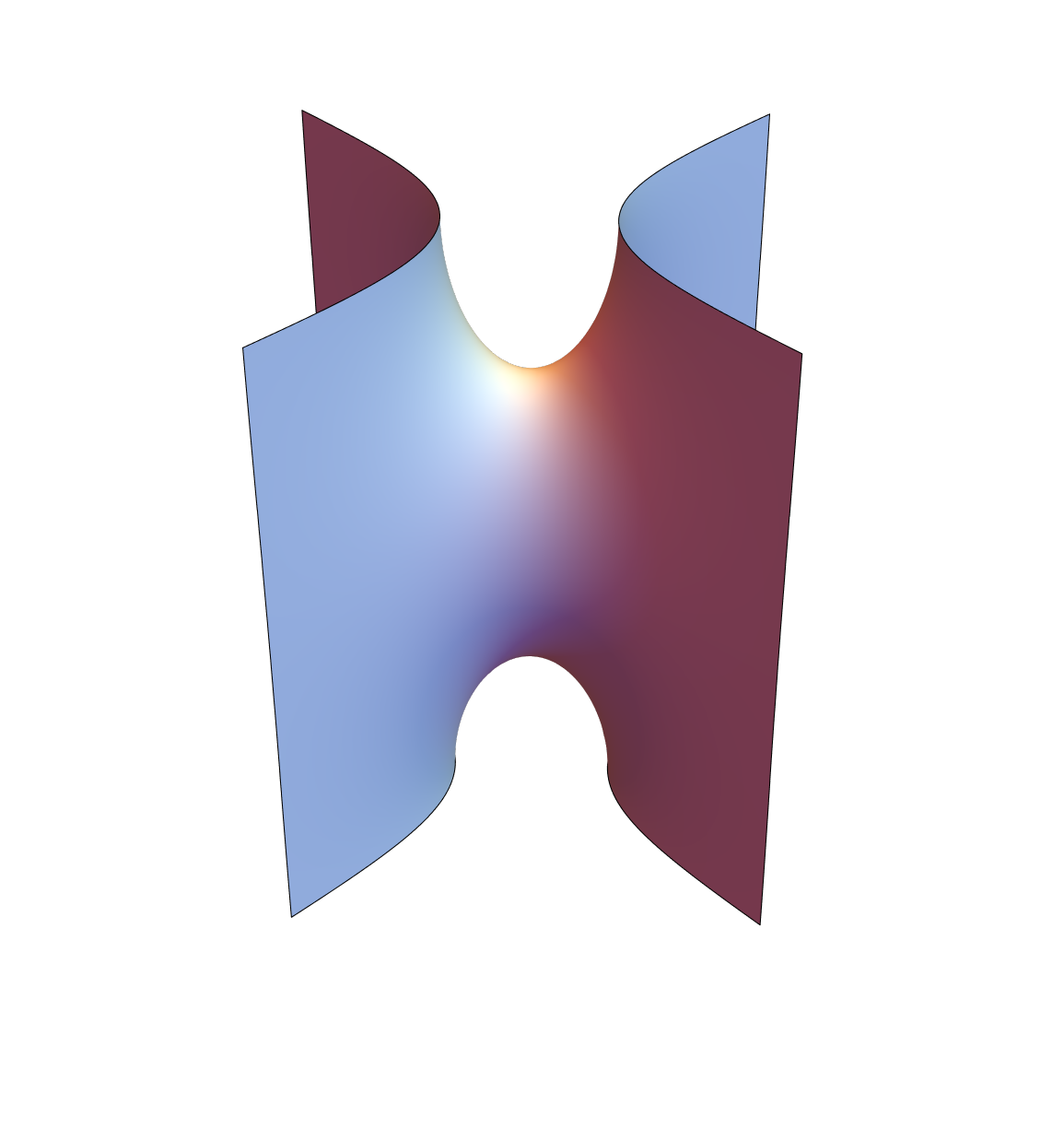}
    \includegraphics[width = .3\textwidth]{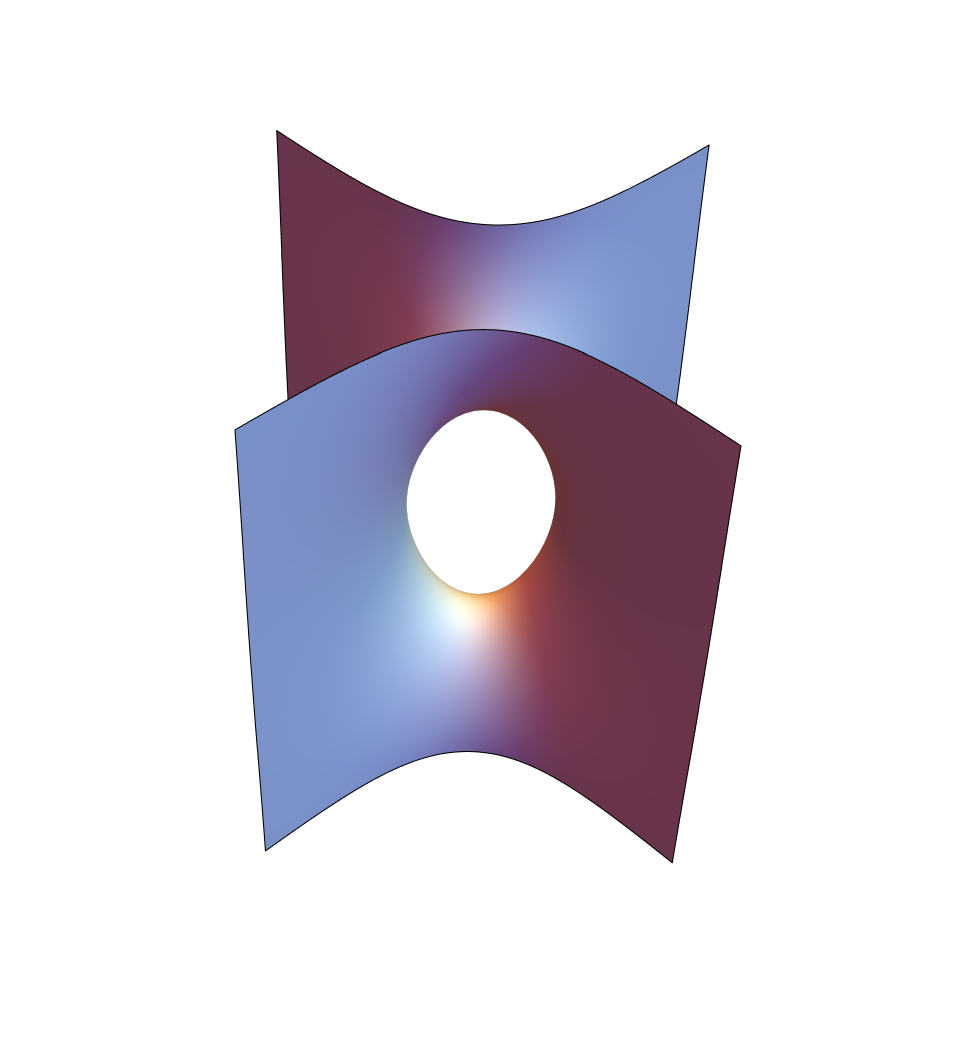}
    \caption{The fundamental domain of the second Scherk surface $\Sch$ embedded in $\mathbf{R}^3$, for the case of orthogonal planes. After the identification $z \sim z + 2\pi$ of the top and bottom edges, the topology of $\Sch$ becomes that of a sphere with four punctures. The cross gets resolved by a smooth transition in the internal space.}
    \label{fig:scherk1}
\end{figure}

Since the origin of the cross's infinite negative mode is its singularity, we might guess that the true bulge desingularizes the cross. In fact, there is a unique extremal surface in $\mathbf{R}^2\times  S^1$ that desingularizes the cross \cite{meeks2005}, and it is called the \textit{second Scherk surface} $\Sch$ \cite{scherk1835bemerkungen}; see Fig.\ \ref{fig:scherk1}.\footnote{In the mathematics literature, the second Scherk surface is usually defined as a minimal surface in $\mathbf{R}^3$, periodic in $z$, and hence the index is infinite. But in our case, we are compactifying the periodic direction and thus the index is finite.} It can be implicitly defined by
\be
\cos{z} = \cos^2{\phi} \cosh \left(\frac{x}{\cos{\phi}}\right) - \sin^2{\phi}\cosh \left(\frac{y}{\sin{\phi}}\right)
\ee
where $\phi$ is the half angle between the asymptotic planes. We can also write a parametric form for the Scherk surface using the Weierstrass--Enneper representation (see appendix \ref{app:WE-rep}). The topology of $\Sch$ is that of a sphere with four punctures, representing the four boundary points delimiting the subregion $\R$. The single scale in the problem is the radius of $ S^1$, and thus the size of the domain of resolution of the cross is parametrically controlled by this scale.

We can again use the trick of the Gauss map $n: \Sch\rightarrow  S^2$ to show that $\Sch$ has index 1. The Gauss map conformally maps the induced metric $h_{\mu\nu}$ onto the round metric on the sphere, $\tilde{h}_{\mu\nu} = e^{2 \omega} {h}_{\mu\nu}$. In this case the Scherk surface is mapped to the sphere with four punctures on the equator, $n(\Sch) =  S^2 \setminus \lbrace{\theta = \pi/2;\varphi = \phi,\pi-\phi,\pi+\phi,2\pi-\phi\rbrace} $, where $(\theta, \varphi)$ are the polar and azimuthal angles respectively and $\phi$ is the half-angle between the asymptotic planes.  We can apply the argument outlined in appendix \ref{app:B} to the Scherk surface and find that it is also index 1. 

Note that the radius $R$ of the $ S^1$, as the only scale in the problem, determines the negative eigenvalue, which is proportional to $1/R^2$. In the limit $R\to0$, the Scherk surface becomes simply the cross in $\mathbf{R}^2$, which as we've noted has index 1 and negative eigenvalue $-\infty$.

\subsubsection{AdS$_{d+1} \times Y$ }

\begin{figure}[h]
    \centering
    \includegraphics[scale = .35]{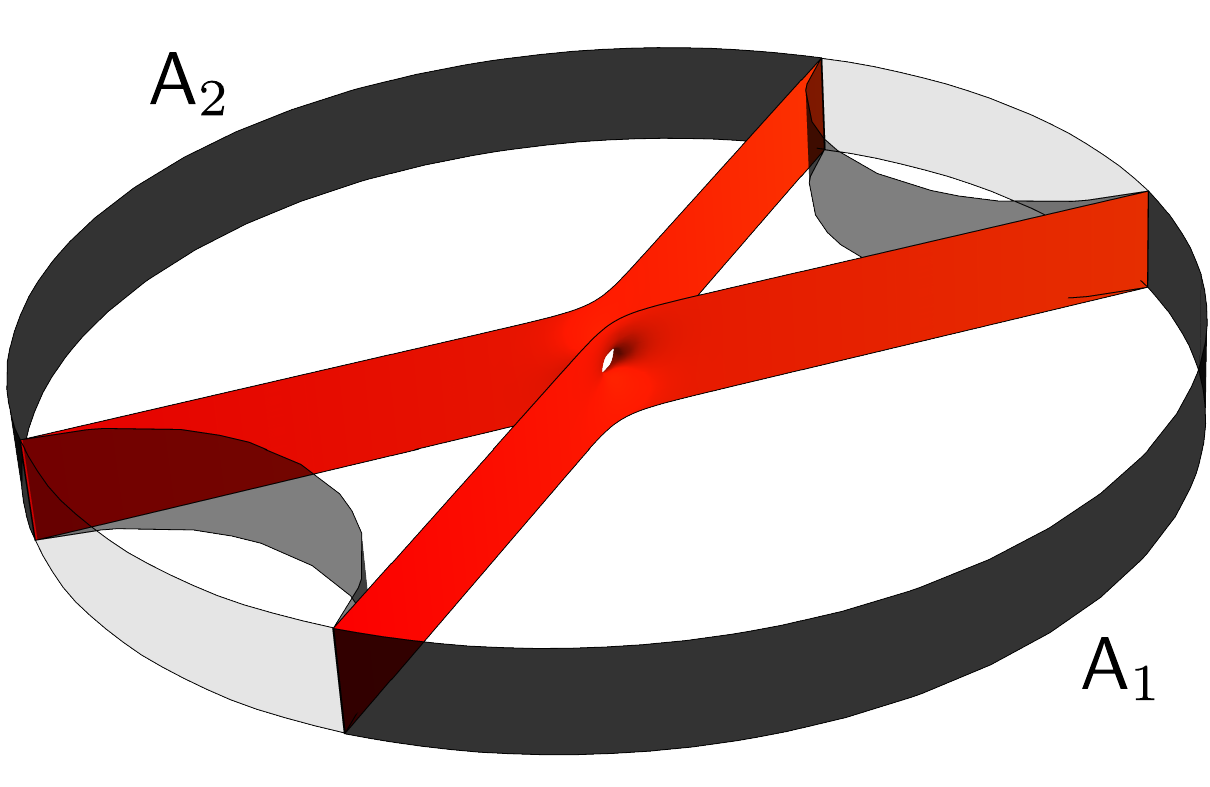}
    \caption{The bulge $X_{\R}^{b}$ in red is a `Scherk surface' $\Sch \subset \mathbf{H}^2\times  S^1$ which does not wrap the internal $ S^1$ and resolves the singular cross. The internal $ S^1$ is deconstructed in the figure, where the top and bottom edges of the strip need to be identified. The characteristic scale of resolution of the cross is set by the asymptotic size of $ S^1$, which in top-down holographic constructions is of the order of the AdS$_3$ scale.
    }
    \label{fig:ads3s1}
\end{figure}

The situation in AdS$_3 \times  S^1$ is expected to be qualitatively similar to the one presented above, with the AdS$_3$ cross resolved into a surface with the same topology as the Scherk surface; see Fig.\ \ref{fig:ads3s1}.

Consider now the more ``realistic'' case of AdS$_3\times  S^3\times {\bf T}^4$, with the $ S^3$ having the same radius as the AdS$_3$ and the ${\bf T}^4$ being much smaller. The cross is certainly not the true bulge, for the reasons given above. On AdS$_3\times  S^3$, it must get resolved into some non-singular surface that is not a product. It would be a very interesting exercise to try to find this surface. Since its characteristic scale and negative eigenvalue are set by the AdS radius, one could then safely ignore the ${\bf T}^4$.

In higher-dimensional AdS spacetimes, the ``naive bulge'' computed ignoring the extra dimensions is not singular and has a finite negative eigenvalue. For example, the eigenvalue $\lambda_0$ of the catenoid bulge in AdS$_4$ is plotted, in units where $\ell_{\rm AdS}=1$, in Fig.\ \ref{fig:negavtiveeigen}. If, in the presence of extra dimensions, there is a KK mode such that the total eigenvalue of the Jacobi operator is negative, then the naive bulge has index greater than 1 and is not the true bulge. No matter the size of the compact space, this will happen for sufficiently large boundary regions, since in the limit $z_0\to0$, $a_-$ goes to 0 and $\lambda_0$ go to $-\infty$. (Even if the naive bulge has index 1, it may not be the true bulge, as there may exist an index-1 non-product surface with smaller area.) It would be interesting to investigate this phenomenon quantitatively and attempt to find the true bulge, for example in the paradigmatic case of the AdS$_{5} \times S^5$ vacuum of the type IIB string, dual to the ground state of four dimensional $\mathcal{N}=4$ SYM on a spatial $S^3$.

\section{Pythons in  
excited states}
\label{sec:excitedstates}

In this section, we qualitatively describe the features of classical pythons in excited states of the holographic system. We discuss three main examples: orbifolds of AdS$_3$, Lin-Lunin-Maldacena (LLM) geometries, and exterior regions of black holes.

\subsection{AdS$_3$ orbifolds}

We start by considering orbifolds of the form AdS$_3/\Gamma$, where $\Gamma$ is a finite subgroup of isometries generated by an elliptic element $g$ of the diagonal $PSL(2,\mathbf{R})$ subgroup of isometries. The elliptic element $g$ has a fixed point in the bulk, which corresponds to a conical defect in the orbifold, of defect angle $\frac{2\pi}{n}$, where $n = |\Gamma|$. 

\begin{figure}[h]
\centering
\includegraphics[width=0.6\textwidth]{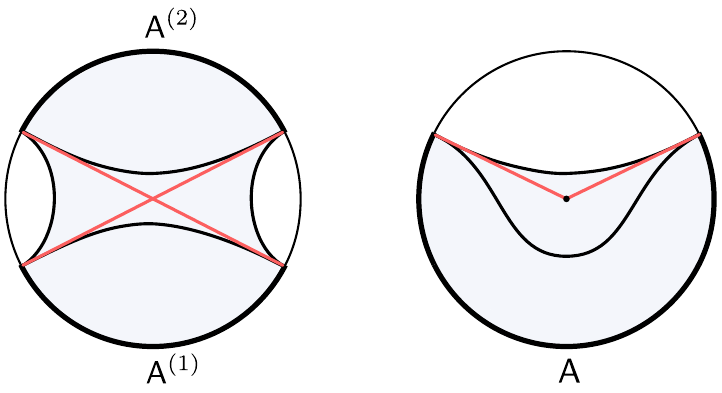}
\caption{On the left, covering space of the $\mathbf{Z}_2$ orbifold. On the right, AdS$_3/\mathbf{Z}_2$ orbifold. The entanglement wedge of $\R$ contains a python. The bulge $X^b_\R$ is the pair of red radial geodesics that intersect the conical defect. }
\label{fig:z2quotient}
\end{figure}

Consider a boundary interval $\R$ comprising more than half of the asymptotic boundary of the orbifolded space. In Fig.\ \ref{fig:z2quotient} we present the case $n=2$, together with the covering space. We see that the entanglement wedge of $\R$ in the orbifold contains a python. The bulge $X^b_\R$ in this case is somewhat peculiar since it consists of two radial geodesics that intersect the conical defect. According to the python's lunch conjecture, the conical defect can be reconstructed, albeit with exponential complexity.

\begin{figure}[h]
\centering
\includegraphics[width=0.7\textwidth]{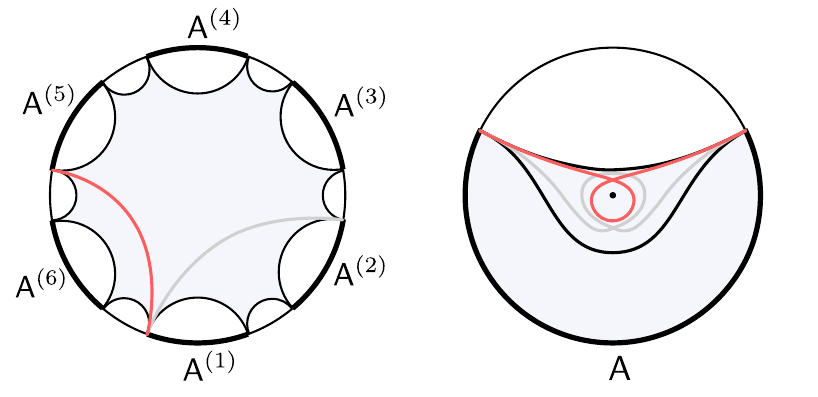}
\caption{On the left, covering space of the $\mathbf{Z}_6$ orbifold. On the right, AdS$_3/\mathbf{Z}_6$. The entanglement wedge of $\R$ contains a python. For $n>3$, there are two geodesics with exactly one self-intersecting point, and thus index 1. The bulge $X^b_\R$ in red corresponds to the smallest among these two geodesics.}
\label{fig:z6quotient}
\end{figure}

For $n\geq 3$ the bulge becomes less singular. In Fig.\ \ref{fig:z6quotient} we present the case $n=6$, together with the covering space. Again, we see that the entanglement wedge of $\R$ in the orbifold contains a python. The bulge $X^b_\R$ is however smooth, and consists of a self-intersecting geodesic that winds around the conical defect. 

In general, for $n\geq 3$, there will be $n$ locally extremal surfaces anchored to the endpoints of a boundary interval $\R$. These correspond to geodesics in the covering  AdS$_3$ space, whose endpoints are related by the action of $\Gamma$. In the orbifolded space, the extremal surfaces with one self-intersecting point will have index $1$ and will therefore be candidates to be the bulge. Given an interval $\R$ of opening angle $\pi<\theta<2\pi$, the constriction, RT, and bulge will correspond to the geodesics in the AdS$_3$ covering space with opening angles $\frac{\theta}{n}$, $\frac{2\pi -\theta}{n}$, and $\frac{2\pi +\theta}{n}$ respectively. The regularized lengths of these geodesics can be found in e.g.\ \cite{Balasubramanian:2014hda} and are given by\footnote{For $n=2$, the length of the bulge is constant, $L(X^b_\R) = 2\ell_{\rm AdS}\log ({2\ell_{\rm AdS}/\varepsilon})$.}
\begin{gather}
 L(X_\R^c) = 2\ell_{\rm AdS}\, \text{log} \left(\frac{2\ell_{\rm AdS}}{\varepsilon} \,\sin(\frac{\theta}{2n}) \right)  , \\
 L(X_\R) = 2\ell_{\rm AdS}\, \text{log} \left(\frac{2\ell_{\rm AdS}}{\varepsilon} \,\sin (\frac{2\pi -\theta}{2n}) \right),\\\
 L(X^b_\R) = 2\ell_{\rm AdS}\, \text{log} \left(\frac{2\ell_{\rm AdS}}{\varepsilon} \,\sin (\frac{2\pi +\theta}{2n})  \right),
\end{gather}
where $\varepsilon$ is a bulk IR regulator. As illustrated in Fig.\ \ref{fig:z6quotient}, it is interesting to note that for $n>3$ there is another index-$1$ extremal surface $X_{\text{index-1}}$ (which as a consequence of Lemma \ref{lem:intersect} of Sec \ref{sec:bulge} must intersect $X^b_\R$), with opening angle $\frac{4\pi -\theta}{n}$ in the covering space, and length corresponding to 
\be 
L(X_{\text{index-1}}) = 2\ell_{\rm AdS}\, \text{log} \left(\frac{2\ell_{\rm AdS}}{\varepsilon} \,\sin (\frac{4\pi -\theta}{2n})\right).
\ee
From Lemma \ref{lem:leastarea} of section \ref{sec:bulge}, the bulge is the minimal among the index-$1$ extremal surfaces. For $n>3$, and $\pi<\theta<2\pi$, $\sin (\frac{4\pi -\theta}{2n}) > \sin (\frac{2\pi +\theta}{2n})$, and therefore this surface is never the bulge, given that $L(X_{\text{index-1}})> L(X^b_\R)$. For $n>4$, there are additional self-intersecting geodesics with index greater than $1$ (due to multiple self-intersecting points) and are thus not candidates to be the bulge.

According to the PLC \eqref{eq:PLC}, the complexity to reconstruct the lunch from $\R$ scales exponentially with the exponent
\be 
\log \mathcal{C}({\mathcal{R}}) \sim  \dfrac{2c}{3} \log \dfrac{\sin (\frac{2\pi +\theta}{2n})}{\sin(\frac{\theta}{2n})}\;,
\ee 
where $c = 3\ell/2G$ is the Virasoro central charge of the dual CFT$_2$.

It is interesting to note that if the region $\R$ is taken to be full boundary, there is no python, given that there is no horizon in the bulk. However, by excising a single point from the boundary ($\theta \rightarrow 2\pi$ in the previous expressions), the entanglement wedge of $\R$ will now contain a python, and the complexity to reconstruct the lunch will grow exponentially, with exponent
\be
\log \mathcal{C}({\mathcal{R}}) \sim \dfrac{2c}{3} \log (2 \cos \frac{\pi}{n})\;.
\ee 
Note that for $n=3$ the complexity is not exponential in $c$, since the exponent vanishes. This can be related to the situation in which we excised three points in section \ref{sec:vacuumads3} from the full AdS$_3$ boundary circle, where we found the same result (in fact, in the covering space the configuration is precisely that of the right Fig.\ \ref{fig:ads3_excised_points}).\footnote{For $n=2$, the complexity is also zero in this limit, but there is no causal shadow. This is easy to see from the configuration in the covering space, where we excise two points (see the left case in Fig.\ \ref{fig:ads3_excised_points}).} For $n>3$, the complexity scales exponentially. A possible interpretation is that, with access to the full boundary, information about the conical defect, such as its deficit angle, can be reconstructed simple operators: these can involve measuring the holonomy of a spatial Wilson loop operator along the boundary $  S^1$. Removing a single point loses access to the loop operator, and thus to simple observables that have access to the conical defect.

\subsection{LLM geometries} 

The next class of states we consider are a subset of the the LLM geometries of \cite{Lin2004}. LLM states are of particular interest because there is a precise and one-to-one map between the boundary states and the bulk geometries, and there is some hope of understanding the bulk-to-boundary map quite explicitly.

The LLM spacetimes are a class of static, 1/2 BPS, asymptotically AdS$_5\times S^5$ solutions of type IIB supergravity. The metric of a constant-time slice can be described as a fibration of $S^3_\alpha\times S^3_\sigma$ over the upper half space of $\mathbf{R}^3$, where $\alpha$ and $\sigma$ are simply labels for the two $S^3$ factors. Let $z$ be one of the coordinates of $\mathbf{R}^3$, with the fibration defined on the $z\ge0$ half. A bounded region of the $z=0$ plane is painted ``black'', with the rest painted ``white''; different choices of black region yield different LLM solutions. As we approach a point on that plane, one of the three-spheres shrinks to zero size; at a black point, the $S^3_\alpha$ shrinks, while at a white point, the $S^3_\sigma$ shrinks. The upper half space is bounded asymptotically by a hemisphere, representing the conformal boundary of the spacetime. The metric, and a description of the dual state, can be found in \cite{Lin2004}.

The simplest example is ${\rm AdS}_5 \times  S^5$ itself, described by a black disc on the $z=0$ plane. Here the $S^5$ is described as $S^3_\sigma$ fibered over a disk, and the spatial slice of AdS$_5$ as $S^3_\alpha$ fibered over a half-line, with the disk and half-line together make up the upper half space.

\begin{figure}[h]
\centering
\includegraphics[width=0.4\textwidth]{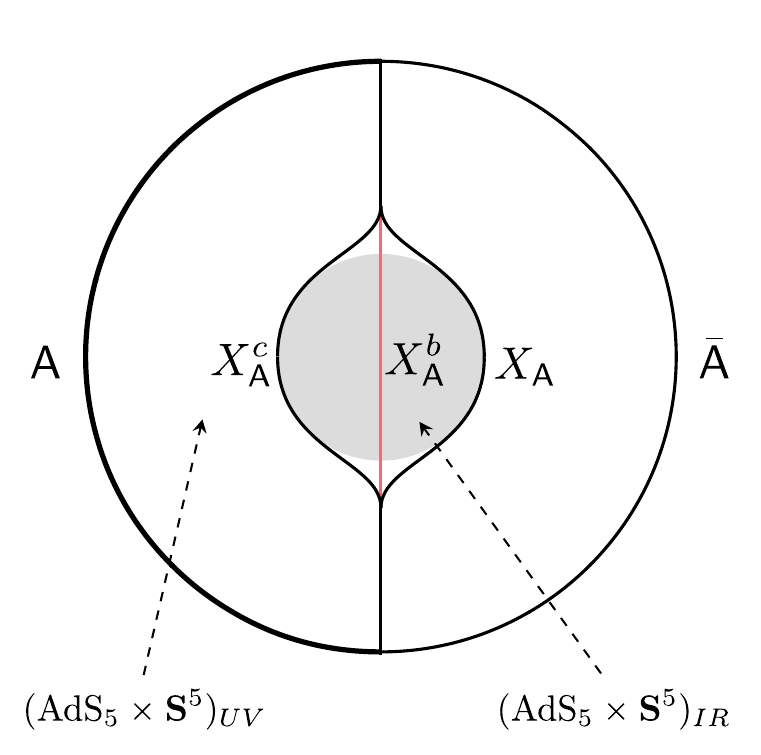}
\caption{Structure of the python on the LLM geometry.}
\label{fig:extremalLLM}
\end{figure}

The next simplest solution is described by a black annulus. A complete description is given in \cite{Lin2004,Balasubramanian:2017hgy}. If the annulus is sufficiently large, the geometry can be shown to interpolate between two versions of  ${\rm AdS}_5 \times  S^5$ with black and white regions (and thus the roles of the $ S^3$s) exchanged, as shown in Fig.\ \ref{fig:extremalLLM}. We denote the ``ultraviolet'' (UV) and ``infrared'' (IR) geometries, based on the relationship between the radial direction of AdS spaces and the energy scale of the dual theory. In the IR, $ S_{\sigma}^3$ is part of the AdS$_5$ and $ S_{\alpha}^3$ is part of the $ S^5$, the opposite of their roles in the UV part of the geometry.

Ref.\ \cite{Balasubramanian:2017hgy} argued that such geometries have an entanglement shadow. In particular they considered the RT surface for half the boundary. Due to the symmetry, there are two minimal surfaces with equal area, passing on either side of the IR region; neither one enters it. In this symmetric situation, it is ambiguous which surface is the RT surface and which one the constriction. There is also an extremal surface that does penetrate this region, namely the symmetric surface that bisects the $S^3_\alpha$ and wraps the other dimensions. It seems likely that these are the only extremal surfaces; if this is true, then Morse theory implies that the symmetric one has index 1.

If we break the symmetry by taking the entangling surface slightly off the equator of the boundary, so that the region $\R$ is slightly smaller than half the boundary, then the closer minimal surface is the smaller one and is the RT surface. In this configuration, there is no python, and the entanglement wedge does not include the IR region.

If we now expand $\R$ from slightly less than half the boundary to slightly more than half, the farther minimal surface becomes the RT surface $X_\R$ and the closer one becomes the constriction $X^c_\R$. The region between them, which includes the entire IR region, is a python. Thus the entanglement wedge jumps from not including the IR region to including all of it, reflecting the ability to begin to recover information about a complex state when one has access to more than half of the degrees of freedom of the theory \cite{Page:1993wv,Page:2013dx}. The symmetric surface presumably persists but is now slightly deformed; assuming again that there are no other extremal surfaces lurking in this geometry, this surface is the bulge $X^b_\R$.

A cartoon of the background and relevant surfaces is shown in Fig.\ \ref{fig:extremalLLM}. The exterior region represents the UV ${\rm AdS}_5\times  S^5$. The shaded region is the IR region encoding the excitation of the vacuum. The extremal surfaces are shown. This situation appears highly reminiscent of the situation for the AdS-Schwarzschild geometry, for which the true RT surfaces sit outside the horizon and an additional surface enters the horizon.  Indeed, the LLM geometries can be considered as a model of black holes; in particular simple local operators give no information as to the detailed structure of the state \cite{Balasubramanian:2005kk,Balasubramanian:2005qu}. 

{\subsection{Eternal black holes}

Another relevant situation arises in this context if we replace the dilute matter by an equilibrium black hole in AdS. Microscopically, we can consider a high-temperature thermofield-double state of two holographic CFTs, or a grand-canonical version of this state. Such a state is semiclassically dual to an eternal black hole in AdS, connecting two different asymptotic boundaries. 

\begin{figure}[h]
\centering
\includegraphics[width=0.3\textwidth]{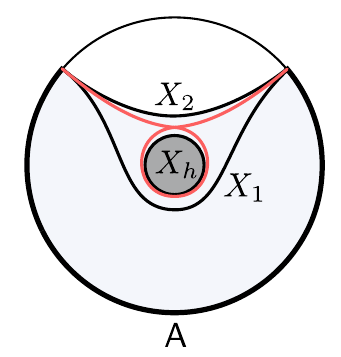}
\hspace{1.5cm}
\includegraphics[width=0.3\textwidth]{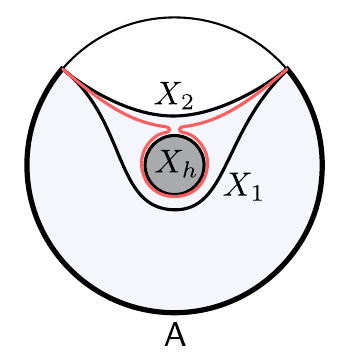}
\caption{Structure of the python on a thermal (or grand-canonical) state of the CFT. { On the left, the structure of the python in $2+1$ dimension, in which the bulge is a self-intersecting geodesic. On the right, the cross section of the higher dimensional python, in which the bulge is smooth and contains a catenoid-like neck.}}
\label{fig:thermalbulge}
\end{figure}

As illustrated in Fig. \ref{fig:thermalbulge}, we consider a disk subregion $\R$ of a single boundary component in this setup. For such an $\R$, there are two competing index-$0$ surfaces in the homology class of $\R$: the connected minimal surface $X_1$ and the disconnected minimal surface $X_{2} \cup X_h$. The compact surface $X_h$ corresponds to the horizon of the black hole. For large enough $\R$, there is a python: the RT surface is the disconnected surface, $X_\R = X_2 \cup X_h$, and the constriction is $X^c_\R = X_1$. As a consequence, there is a bulge surface $X_\R^b$ of index $1$ in between them. In $2+1$ dimensions the bulge corresponds to a self-intersecting geodesic \cite{Brown:2019rox}, while in higher dimensions, or when the compact dimensions are considered, the bulge $X_\R^b$ is a smooth surface; in higher dimensions it has a catenoid-like neck \cite{Hubeny:2013gta}. As $\R$ is taken to be the entire boundary, the python dissapears. In this limit, the entanglement wedge $\rew$ becomes the exterior of the black hole, which can be reconstructed in a simple way via appropriately smeared HKLL operators of the boundary.
}

\section{Pythons in the black hole interior}
\label{sec:bhint}

In this section, we will study states of multiple black holes with shared geometric interiors. In the holographic system, these correspond to entangled states of independent CFTs. For the purpose of this section, we will restrict to subregions $\R$ made up of entire compact boundary components.

\subsection{Two-boundary wormhole}\label{sec:two_bdry}

\begin{figure}[h]
    \centering
    \includegraphics[width = .55\textwidth]{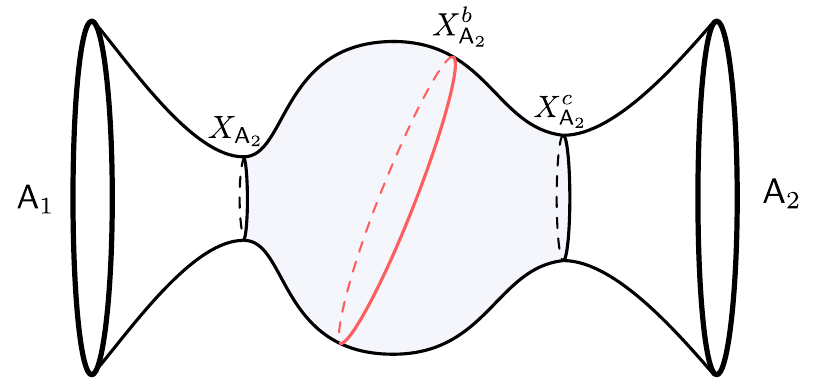}
    \caption{A two-sided black hole microstate with a python.}
    \label{fig:twosidedpython}
\end{figure}

We start by considering a bipartite state $\ket{\Psi} \in \mathcal{H}_1\otimes \mathcal{H}_2$ whose bulk description contains the conformal boundary $\R_1 \sqcup \R_2 $ and a two-sided python geometry on the time-symmetric slice, shown in Fig.\ \ref{fig:twosidedpython}.\footnote{Explicit examples of such states include the so-called {\it partially entangled thermal states} (PETS) in holographic CFTs \cite{Goel:2018ubv,Balasubramanian:2022gmo,Balasubramanian:2022lnw,Antonini:2023hdh,deBoer:2023vsm}.} Let us select the boundary component $\R_2$ as our subregion. We can assume without loss of generality that $\R_2$ contains the black hole interior in its entanglement wedge, as indicated in Fig.\ \ref{fig:twosidedpython}. The constriction $X^c_{\R_2}$ is simply the apparent horizon of the right black hole, while the RT surface $X_{\R_2} = X_{\R_1}$ is the apparent horizon of the left black hole. In the interior of the python region $\tilde{N}$ in between these two surfaces, there exists an index-$1$ bulge surface $X^b_{\R_2}$. 

According to the python's lunch conjecture, the amount of post-selection required to access the interior from $\R_2$ scales with the exponent
\be
\log \Co(\mathcal{R}_2) \sim \frac{1}{2}\left(S_{\text{gen}}(X^b_{\R_2})-S_{\text{gen}}(X^c_{\R_2})\right).
\ee

Next we consider applying the conjecture to the full boundary $\R = \R_1\R_2$. The RT surface of the full boundary is clearly empty, $X_\R = \emptyset$, which represents the fact that the two-sided state is pure. The constriction for $\R$ is simply the union of the two apparent horizons, $X_\R^c = X_{\R_2}\cup  X^c_{\R_2}$. At first we might guess that the bulge is just twice the $\R_2$ bulge, $2X_{\R_2}^b$; however, given that the index is additive under disjoint union, this surface has index 2, and is therefore not a candidate for the bulge. According to Lemma \ref{lem:leastarea} of section \ \ref{sec:bulge}, assuming that there is no other minimal surface in the interior, the true bulge is the minimal index-$1$ surface homologous to the entire boundary (or, equivalently, null-homologous). This must include $X^b_{\R_2}$, the only index-1 extremal surface; in order to obey the homology constraint, we must add another index-0 surface, the smallest of which is the apparent horizon of the left black hole; in all, we find $X^b_\R = X^b_{\R_2} \cup X_{\R_2}$. This implies
\be 
S_\text{gen}(X_\R^b) - S_\text{gen}(X_\R^c) = S_\text{gen}(X_{\R_2}^b) - S_\text{gen}(X_{\R_2}^c)\;.
\ee 
Therefore, according to the PLC \eqref{eq:PLC}, the 
complexities of reconstructing the interior region using just $\R_2$ or all of $\R$ are the same:
\be\label{eq:compsame}
\Co(\mathcal{R}_{12}) = \Co(\mathcal{R}_2)\;,
\ee 
where $\mathcal{R}_{12}$ is the recovery channel of the global bulk-to-boundary map $\mathcal{N}(\rho_{\rew})=  V \rho_\rew V^\dagger$, while $\mathcal{R}_2$ is the recovery channel of the restricted map $\mathcal{N}'(\rho_{\rew'})=  \text{Tr}_{\R_1}(V \rho_{\rew'}\otimes{\sigma_{\bar{\rew'}}} V^\dagger)$, which only involves acting on $\R_2$. Both recovery channels in \eqref{eq:compsame} need to be understood as restricted to states that differ in the interior region $\tilde{N} \subset \rew \cap \rew'$, given that the entanglement wedge $\rew$ strictly contains $\rew'$.

{

As remarked in the introduction, we find that this behavior of the  
complexity is in contrast to the one expected for the class of two-sided states motivating the PLC in \cite{Brown:2019rox}. The states in \cite{Brown:2019rox} were by assumption obtained dynamically from the unitary time evolution coupling the two subsystems, starting from a simple entangled state without a bulge and evolving it for a relatively short time. For those states, with access to the full boundary, the reconstruction was by assumption simple; one just needed to undo the unitary time-evolution. In that case, the individual tensors of the tensor network modelling the semiclassical state are correlated to produce the effect that, when $\R$ is taken to be the full boundary, the python effectively dissapears.

In contrast, in this case of two-sided geometric states, using the PLC for the full boundary, we find that the 
reconstruction using both sides is just as complex as the one restricted to subsystem $\R_2$. Applying the conjecture to $\R_1\R_2$ amounts to the implicit assumption that the tensor network structure modelling the geometry is generic, so different tensors at different points appear uncorrelated, at least at leading order. This makes the bulge for two-sided geometric microstates not go away even when the full system $\R_1 \R_2$ is considered.  The assumption of uncorrelated tensors is reasonable since, after all, these states are not prepared unitarily by the collapse of matter, and generic tensor networks with large bond dimension are successful in capturing the RT formula \cite{Hayden:2016cfa}.

What is really peculiar is that, according to the PLC \eqref{eq:PLC}, there is no computational advantage at all in adding the second boundary $\R_1$ in order to reconstruct the black hole interior. This is certainly a non-trivial feature of the conjecture applied to these states, given that, with access to both boundaries, there are known simple bi-local holographic probes that contain some information about the interior. An example is the EPR correlation function $\bra{\Psi}\mathcal{O}_{\R_1}\mathcal{O}_{\R_2}\ket{\Psi}$ of heavy scalar primary $\mathcal{O}$, which presumably contains information about the length of the wormhole. However, according to the tensor network/geometry intuition, these probes will not help to decode the local physics of the lunch, at least deep in the black hole interior. The optimal way to reconstruct the lunch is to follow a minimax level set path, which involves leaving subsystem $\R_1$ intact, and overcoming the postselection solely acting on $\R_2$. 
}

\subsection{Multi-boundary wormhole}

We will now show that this feature of the PLC extends to multi-boundary states with connected wormholes. To be concrete, we consider time reflection-symmetric microstates of a family of three dimensional black holes. Microscopically, the states in question live in the Hilbert space of $n$ copies of a putative holographic two-dimensional CFT, $\ket{\Psi} \in \mathcal{H}_1\otimes \ldots\otimes \mathcal{H}_n$. The states possess a semiclassical description, with initial data specified at the moment of time-symmetry of the spacetime, $\Sigma^n_g$, which consists of a Riemann surface of genus $g$ and $n$ boundaries. We now give a brief summary of how these states are constructed (see e.g. \cite{PhysRevD.53.R4133,Brill:1998pr,Krasnov:2000zq,Skenderis:2009ju} for details).\footnote{Some of these states are prepared via suitable Euclidean CFT path integrals. However, we will not worry on how they are prepared, simply we view them as valid state vectors in Hilbert space.}

The constraint equations of AdS$_3$ gravity require the metric on $\Sigma^n_g$ to be the unique constant negative curvature metric on this Riemann surface. Mathematically, the standard way to construct this metric is to uniformize the Riemann surface on the hyperbolic disk $\Sigma^n_g = \mathbf{H}^2/\Gamma$ via a discrete subgroup $\Gamma$ of $PSL(2,\mathbf{R})$ isometries, generated freely by hyperbolic elements, which leads to a smooth bulk metric. The group $\Gamma$ is known as a Fuchsian group of the second kind (for $n=0$ this reduces to the standard Fuchsian uniformization of compact Riemann surfaces of genus $g>1$). 

The Lorentzian evolution of this initial data can be constructed as follows. Observe that $PSL(2,\mathbf{R})$ can be extended as the diagonal subgroup of $SO(2,2) \simeq SL(2,\mathbf{R})_L\times SL(2,\mathbf{R})_R/\mathbf{Z}_2$ of simply connected isometries of AdS$_3$. This is essentially the subgroup which commutes with time-reflection symmetry and thus preserves the time reflection-symmetric hyperbolic slice which uniformizes $\Sigma^n_g$. Therefore, a complete vacuum spacetime solution can be constructed simply from the quotient $\widehat{\text{AdS}_3}/\Gamma$, which respects the metric of $\Sigma^n_g$ at the time-reflection symmetric slice. Here $\widehat{\text{AdS}_3}$ represents the causal wedge in AdS$_3$ of the time reflection-symmetric boundary circle, where the fixed points of $\Gamma$ have been removed. Specifically, the spacetime metric can be written locally in FRLW coordinates, as
\be 
\text{d}s^2 = -\text{d}t^2 + \cos^2\left(t/\ell_{\rm AdS}\right) \, (\text{d}\Sigma_g^n)^2\;,
\ee 
where $(\text{d}\Sigma_g^n)^2$ is the constant negative curvature metric on $\Sigma_g^n$. For more details, we refer to the reader to \cite{PhysRevD.53.R4133,Brill:1998pr,Krasnov:2000zq,Skenderis:2009ju,Balasubramanian:2014hda,Maxfield:2014kra,Maloney:2015ina} and references therein.

The simplest family of states constructed this way is $\ket{\Sigma^2_0}$, which are specified by a hyperbolic Riemann surface $\Sigma^2_0$ with annulus topology. This surface is uniformized as $\Sigma^2_0= \mathbf{H}^2/\Gamma$ by a Fuchsian group $\Gamma$ generated by a single hyperbolic element $\Gamma = \langle g \rangle$. The annulus $\Sigma^2_0$ arises naturally as the fundamental domain of $\Gamma$. The metric in this case can be written explicitly, 
\be 
\text{d}s^2|_{\Sigma^2_0} = \text{d}\rho^2 + \frac{L^2}{(2\pi)^2} \cosh^2(\rho /\ell_{\rm AdS}) \text{d}\phi^2\;,
\ee
for $\rho \in \mathbf{R}$ and $\phi \in [0,2\pi)$. This is simply the initial data of the BTZ black hole, obtained in a different way. These states contain a single modulus, the length $L$ of the horizon, which is the minimal closed geodesic on $\Sigma^2_0$.\footnote{Additionally, there is an infinite set of closed geodesics, of lengths $L_k = |k|L$, labelled by the winding number around the horizon, which is in one-to-one correspondence with the conjugacy classes of $\Gamma$, that is, $[g^k]$ for $k\in  \mathbf{Z}$.} The length of the horizon $L$ naturally determines the ADM energies of the corresponding state. In this case, these states do not contain pythons --- the Einstein-Rosen bridge has vanishing volume on $\Sigma^2_0$.

For $n> 2$ or $g> 0$, the rest of the states include pythons in the black hole interior. For given values of $(n,g)$ the states are parametrized by moduli in (discrete quotients of) the Teichm\"{u}ller space $\mathcal{T}_{g,n}$ of Riemann surfaces of genus $g$ and $n$ boundaries. The natural way to parametrize the moduli space of $\Sigma_g^n$ is to cut the python region bounded by closed geodesics into pairs of pants, and use the so-called Fenchel-Nielsen coordinates to glue these pairs of pants together. In this way, it is easy to see that total number of moduli is $6g-6+3n$. In our case of concern, we will fix the length of each apparent horizon $L_i$ for $i=1,...,n$, and in this way fix the coarse-grained entropy of each boundary. Moreover, for the purpose of this discussion, we will restrict to $g=0$ and $n> 2$, where the number of additional moduli is $2n-6$.

\subsection{Three-boundary wormhole}

\begin{figure}[h]
    \centering
    \includegraphics[width = .45\textwidth]{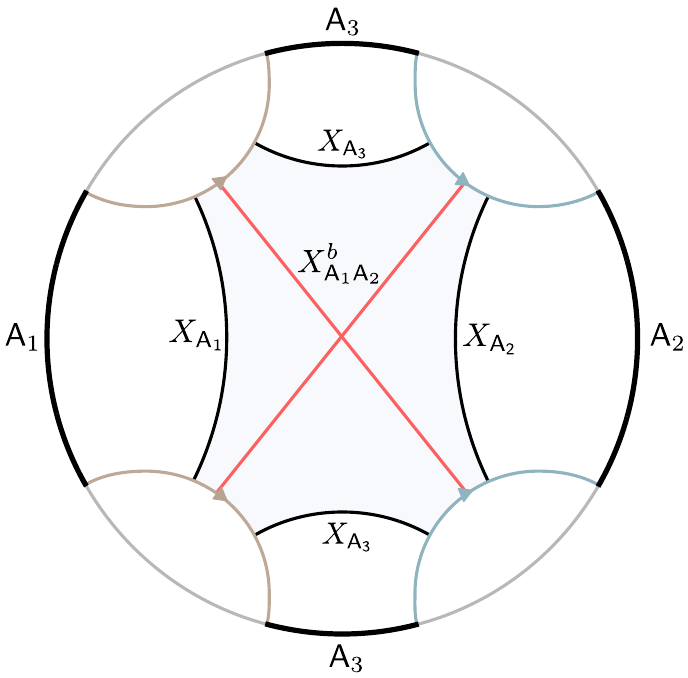}
    \hspace{1.5cm}
    \includegraphics[width = .4\textwidth]{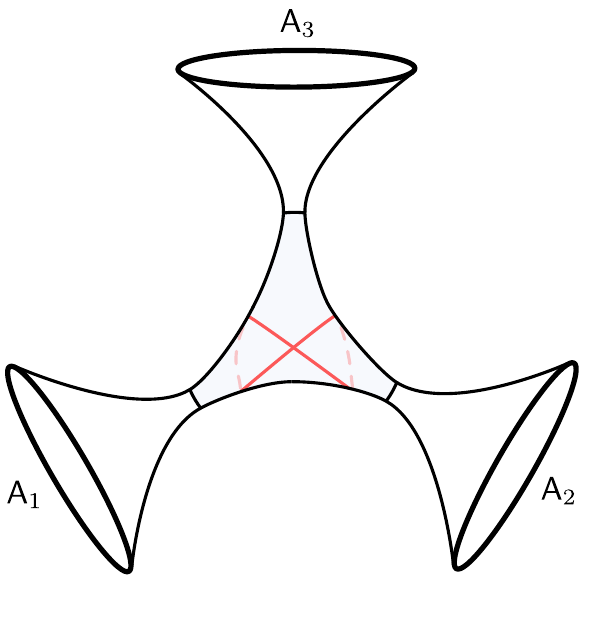}
    \caption{On the left, fundamental domain of the Fuchsian group $\Gamma$ generated by two elements uniformizing the three-boundary wormhole $\Sigma^3_0$ with equal horizon moduli, $L_1 = L_2 = L_3$. On the right, illustration of the three-boundary wormhole.}
    \label{fig:3bdy}
\end{figure}

We now consider three-boundary wormhole states $\ket{\Sigma^3_0}$. In order to motivate the general claim, let us set $L_1=L_2=L_3$ for the time being. The corresponding states have discrete $\text{Sym}(3)$ permutation symmetry. Fixing the lengths of the horizons determines the moduli of $\Sigma^3_0$ completely. A representation of the fundamental domain of $\Gamma$ is shown in Fig.\ \ref{fig:3bdy}. 

Index-$1$ extremal surfaces on $\Sigma^3_0$ consist of self-intersecting geodesics with exactly one crossing point, where crossing is only permitted between neighbouring geodesic segments. As for the case of vacuum AdS$_3$ of section \ref{sec:vacuumpython}, it is only in this case that there are just two deformations that resolve the cross and reduce the area of the surface, so the surface has index 1. For multiple geodesic segments meeting at the same point, or for multiple crossing points, there are more than two area-reducing deformations, so the index is always greater than 1. This leaves little topological freedom to determine what the index-1 surfaces are in $\Sigma^3_0$. Namely, bulge candidates will necessarily wrap around two of the constrictions. In Fig.\ \ref{fig:3bdy} we have represented $X^b_{\R_1\R_2}$, the extremal index-$1$ surface homologous to $\R_1\R_2$.

Given this, we consider different boundary subregions $\R$, consisting of different numbers of connected boundary components. We list all of the relevant extremal surfaces for the PLC in table \ref{tab:3bdtable}, which up to permutations, determine all the boundary-homologous pythons for these states. 
\begin{table}[h]
    \centering
    \begin{tabular}{|c|c|c|c|c|}
        \hline 
        $\R$ & RT  &  constriction & bulge & $8 G \log \mathcal{C}(\mathcal{R})$\\ \hline
        $\R_1$ & $X_{\R_1}$ & -  & - & - \\ \hline
        $\R_1 \R_2$ & $X_{\R_3}$ & $X_{\R_1}\cup X_{\R_2}$  & $X^b_{\R_1\R_2}$ & $L_{12}- L_3$\\ \hline
        $\R_1  \R_2  \R_3$ & $\emptyset$ & $X_{\R_1}\cup X_{\R_2} \cup X_{\R_3}$  & $X^b_{\R_1\R_2} \cup X_{\R_3}$ & $L_{12}- L_3$ \\ \hline
    \end{tabular}
    \caption{Relevant extremal surfaces in the python for each boundary region. For $\R = \R_1 \R_2$ there are two additional bulge candidates: $X^b_{\R_1\R_3} \cup  X_{\R_1}$ and $X^b_{\R_2\R_3} \cup X_{\R_2}$. Both of these surfaces have larger area than $X^b_{\R_1\R_2}$, and are therefore not the bulge of the python according to Lemma 2 of section  \ref{sec:bulge}. For $\R = \R_1\R_2\R_3$, there are five additional candidate bulge surfaces: $X^b_{\R_1\R_3} \cup X_{\R_2}$ and $X^b_{\R_2\R_3}\cup X_{\R_1}$, with the same area as $X^b_{\R_1\R_2} \cup X_{\R_3}$, so any of these three surfaces can be considered equally. Moreover, the index-1 surfaces $X^b_{\R_1\R_2}\cup X_{\R_1}\cup X_{\R_2}$, $X^b_{\R_1\R_3}\cup X_{\R_1}\cup X_{\R_3}$ and $X^b_{\R_2\R_3}\cup X_{\R_2}\cup X_{\R_3}$ have larger area, and thence they do not correspond to the bulge.}
    \label{tab:3bdtable}
\end{table}

We observe that, in accordance with the discussion of section \ref{sec:breakisom}, the bulge for the full boundary breaks the $\text{Sym}(3)$ permutation symmetry of the state $\ket{\Sigma_0^3}$ into a $\text{Sym}(2)$ subgroup. We also evaluated the 
complexity to reconstruct the lunch for each subregion in table \ref{tab:3bdtable}. We observe that the bulge for the three boundaries is the union of the bulge for the two boundaries and the constriction of $\R_3$. This implies that the complexity of reconstructing the lunch is independent of whether $\R_3$ is used or not,  
    \be\label{eq:3bdycomplexity}
    \Co(\mathcal{R}_{123}) = \Co(\mathcal{R}_{12}) \;,
    \ee 
where $\mathcal{R}_{123}$ is the global, unrestricted recovery map, while $\mathcal{R}_{12}$ is the map restricted to $\R_1\R_2$. The minimax foliation indicates that the optimal way to reconstruct the lunch is to leave $\R_3$ intact, and to only act on $\R_{1}\R_{2}$. 

We now want to study how this feature generalizes to the case in which $L_1,L_2,L_3$ are three general moduli. To do this we will use a one-to-one correspondence between closed geodesics in $\Sigma_{g}^n$ and conjugacy classes in $\Gamma$. The length of the geodesic associated to the conjugacy class of the group element $g$ is determined by
\be\label{eq:lengthg}
L = 2\cosh^{-1}\left|\dfrac{\text{Tr}g}{2}\right|\;,
\ee 
in the representation $g = \begin{psmallmatrix} a&b\\c&d\end{psmallmatrix}$ of the $PSL(2,\mathbf{R})$ isometry, acting by fractional linear transformations $z\rightarrow \frac{az+b}{cz+d}$, of the Poincar\'{e} upper half plane model of $\mathbf{H}^2$. Namely, this is the geodesic that connects the two fixed points of $g$ at the $\text{Im}(z) =0$ boundary, which becomes closed, and stays smooth up to crossing points, by virtue of the fact that $g$ acts freely and properly discontinuously in $\mathbf{H}^2$, for any $g \in \Gamma$. 

The group $\Gamma$ for the three-boundary wormhole is freely generated by the elements
\be 
g_1 = \left(\begin{matrix} \cosh \frac{L_1}{2}&\sinh \frac{L_1}{2}\\\sinh \frac{L_1}{2}&\cosh \frac{L_1}{2}\end{matrix}\right) \hspace{.5cm}, \hspace{.5cm}g_2 = \left(\begin{matrix} \cosh \frac{L_2}{2}&e^{\alpha}\sinh \frac{L_2}{2}\\e^{-\alpha}\sinh \frac{L_2}{2}&\cosh \frac{L_2}{2}\end{matrix}\right)\hspace{.5cm}\;.
\ee 
The parameter $\alpha$ controls the separation between the $g_1$ and $g_2$ semicircles delimiting the fundamental domain of $\Gamma$. For these circles not to intersect each other, we must impose the constraint
\be\label{eq:conditionthreebdy}
e^{\alpha}> \coth\frac{L_1}{4} \coth\frac{L_2}{4}\;.
\ee 

Given this choice of generators, the third horizon is associated to the conjugacy class of the group element $g = g_1 g_2^{-1} \in \Gamma$. From \eqref{eq:lengthg} the length of the third horizon is determined in terms of the previous three moduli,
\be\label{eq:length3threebdy} 
L_3 = 2 \cosh^{-1}\left(\cosh \alpha \sinh \frac{L_1}{2} \sinh \frac{L_2}{2}-\cosh \frac{L_1}{2} \cosh \frac{L_2}{2} \right).
\ee 
It is possible to check that \eqref{eq:conditionthreebdy} is equivalent to the condition $L_3 >0$.

The index-1 surface $X^b_{\R_1\R_2}$ will correspond to the closed (self-intersecting) geodesic associated to the conjugacy class of the group element $g = g_1 g_2 \in \Gamma$. Using \eqref{eq:lengthg} again, we arrive at the length of the index-1 surface
\be 
L_{12} = 2 \cosh^{-1}\left(\cosh \alpha \sinh \frac{L_1}{2} \sinh \frac{L_2}{2}+\cosh \frac{L_1}{2} \cosh \frac{L_2}{2} \right).
\ee 
Substituting $\cosh \alpha$ in terms of the length of the horizons from \eqref{eq:length3threebdy}, we obtain the simple relation for the length of $X^b_{\R_1\R_2}$ in terms of the three horizon lengths
\be 
\cosh \frac{L_{12}}{2} = 2 \cosh \frac{L_1}{2} \cosh \frac{L_2}{2} + \cosh \frac{L_3}{2}\;.
\ee 
The expressions for $L_{13}$ and $L_{23}$ follow from simple permutation of the indices in this formula. It is easy to check that the length of the index-1 surfaces in the lunch is always larger than the horizon lengths, $L_{ij}> L_k \;\forall\; i,j,k \in \lbrace 1,2,3\rbrace$ with $i<j$.

Consider $L_1 \leq L_2 \leq L_3$ without loss of generality. In this case, the three index-1 surfaces satisfy 
\be
L_{12}\leq L_{13}\leq L_{23}\,.
\ee

Moreover, assume that $L_3> L_1 + L_2$, so that $\R_3$ can itself access the interior.\footnote{In the opposite case where $L_3< L_1 + L_2$, the interior can only be accessed with two boundaries and the complexity to reconstruct the lunch with the complete holographic system $\R_1\R_2\R_3$, $\Co(\mathcal{R}_{123})$, is simply the minimal amongst the complexities of reconstructing it with two boundaries, i.e. $\Co(\mathcal{R}_{123}) = \text{min}\lbrace \Co(\mathcal{R}_{12}),\Co(\mathcal{R}_{13}),\Co(\mathcal{R}_{23})\rbrace$. Any of the three quantities can be minimal, depending on the moduli.} In this case it is straightforward to verify the following relations
\begin{gather}
L_{13}\geq L_{12}+L_1 \Leftrightarrow L_3 \geq  L^+_{12}\;, \label{eq:ineq2}\\
L_{23} \geq L_{12}+L_2 \Leftrightarrow L_3 \geq L^+_{21}\label{eq:ineq3}\;,
\end{gather}
where 
\begin{multline}
L^+_{ij} = \\ 2\cosh^{-1}\left[\cosh \frac{L_i}{2}\cosh \frac{L_j}{2} \left(2\cosh L_i -1 + \sqrt{\coth^2\frac{L_i}{2}\coth^2\frac{L_j}{2} + 4\cosh L_i (\cosh L_i -1)}\right) \right] .
\end{multline}
It is also easy to show that $L_{1}+L_2\leq L_{12}^+ \leq L_{21}^+$.

\begin{table}[h]
    \centering
    \begin{tabular}{|c|c|c|c|c|}
        \hline 
        $\R$ & RT  &  constriction & bulge & $8 G\log \Co(\mathcal{R})$\\ \hline
        $\R_3$ & $X_{\R_1}\cup X_{\R_2}$ &  $X_{\R_3}$ &  $X^b_{\R_1\R_2}$ & $L_{12}-L_3$\\ \hline
        $\R_1\R_3$ & $X_{\R_2}$ &  $X_{\R_1}\cup X_{\R_3}$ &  $X^b_{\R_1\R_2} \cup X_{\R_1}$& $L_{12}-L_3$\\ \hline
        $\R_2\R_3$ & $X_{\R_1}$ & $X_{\R_2}\cup X_{\R_3}$  & $X^b_{\R_1\R_2} \cup X_{\R_2}$& $L_{12}-L_3$\\ \hline
        $\R_1 \R_2\R_3$ & $\emptyset$ & $X_{\R_1}\cup X_{\R_2} \cup X_{\R_3}$  & $X^b_{\R_1\R_2} \cup X_{\R_1} \cup X_{\R_2}$ & $L_{12}-L_3$\\ \hline
    \end{tabular}
    \caption{Relevant extremal surfaces in the python for each boundary region for the general case of different moduli with $L_3 > L_{21}^+$. The complexity to reconstruct the lunch is constant for any choice of $\R$.}
    \label{tab:3bdtableasym}
\end{table}

Assume that $L_3 \geq L_{21}^+$. Given \eqref{eq:ineq2} and \eqref{eq:ineq3}, we can directly evaluate the bulge for all of the boundary subregions. In table \ref{tab:3bdtableasym}, we list all possible pythons, with their relevant extremal surfaces and the complexity to reconstruct the lunch, for any choice of $\R$. Again, we find that when $L_3 \geq L_{21}^+$, the complexity to reconstruct the interior is independent of whether $\R_1$ or $\R_2$ is included in the reconstruction,
\be
\Co(\mathcal{R}_{123}) = \Co(\mathcal{R}_{13}) = \Co(\mathcal{R}_{23}) = \Co(\mathcal{R}_{3}) \;,
\ee
since the minimax foliation of the lunch leaves the horizons $X_{\R_1}$ and $X_{\R_2}$ intact. In the limit $L_1,L_2,L_3\rightarrow \infty$, formulated in terms of generalized entropies, the condition $L_3 \geq L_{21}^+$ becomes
\be 
S_{\text{gen}}(X_{\R_3})\gtrsim S_{\text{gen}}(X_{\R_1}) + 3S_{\text{gen}}(X_{\R_2})\,
\ee 
Note that this means that the fraction of the coarse-grained entropy carried by $\R_3$ must be at least $\frac{2}{3}$ of the total entropy in order for this effect to take place. In that particular case $\R_1$ and $\R_2$ each carry $\frac{1}{6}$ of the total coarse grained entropy of the state.\footnote{In the pinching limit $L_1\rightarrow 0$ one does not recover the result for the two-boundary lunch since all of the bulges in table \ref{tab:3bdtableasym} wrap the horizon $X_{\R_1}$.}

In general, for $L_3 \leq L_{21}^+$, the complexity will not completely plateau as a function of the number of boundaries in $\R$. This means that the reconstruction of the interior with three boundaries will be strictly simpler than with $\R_3$, i.e. $\mathcal{C}(\mathcal{R}_{123}) < \mathcal{C}(\mathcal{R}_{3})$. However, due to the topological constraints of the index-1 surfaces, it will always be true that the simplest way to reconstruct the lunch with two boundaries will be the optimal way to reconstruct the lunch with three,
\be 
\Co(\mathcal{R}_{123}) = \text{min}\lbrace \Co(\mathcal{R}_{12}),\Co(\mathcal{R}_{13}),\Co(\mathcal{R}_{23})\rbrace\;.
\ee 

These results suggest the following generalization: for geometric black hole microstates, the complexity to reconstruct the interior plateaus after a certain amount of entropy is included, in the form of a single black hole. After this point, adding more black holes into the boundary system in order to reconstruct their shared interior does not help --- the optimal reconstruction leaves these additional black holes intact. It is tempting to conjecture that, for an $n$-boundary wormhole, the single black hole has an entropy at least a fraction $\frac{n-1}{n}$ of the total entropy in order for the complexity to plateau. This includes the case of the PETS for $n=2$. We shall not attempt to provide a proof of this in this paper.

The complexity plateau phenomenon discussed in this section is closely analogous to the non-extensivity of the log-complexity for black branes discussed in subsection \ref{sec:planar} above. In both cases, the effect occurs because, for sufficiently large subsystems, the bulge coincides (exactly or approximately) with the constriction. We expect this to be a general phenomenon, and points to a surprising aspect of the complexity of reconstruction.

\subsection{Generalization for $n$-boundary wormhole}

For $n>3$, the multi-boundary wormhole states $\ket{\Sigma_0^n}$ include a landscape of closed minimal surfaces in the black hole interior. There are also additional index-1 surfaces that wrap more than two horizons at the same time. This makes the situation vastly more complicated given that the python includes multiple lunches. Moreover, in these cases, the minimal surfaces in the interior generally intersect each other. 

Therefore, we must generalize the procedure specified in \eqref{eq:PLC2} to determine the complexity to reconstruct the lunch, to include situations where multiple choices of the set $\mathcal{S}$ of non-intersecting adjacent minimal surfaces are possible. The proposed generalization is to find the foliation with the least amount of postselection from the following steps:
\begin{enumerate}
    \item Select a set $\mathcal{S}$ of non-interesecting homologous adjacent minimal surfaces in $N$ such that $X_\R, X_\R^c \in \mathcal{S}$.
    \item Find the bulges $X^{b,i}_\R$ as the maximin surfaces in between each pair of adjacent minimal surfaces in $\mathcal{S}$.
    \item Compute the complexity of reconstructing the lunch given the discrete set of bulges and minimal surfaces associated to this choice of $\mathcal{S}$, namely using \eqref{eq:PLC2}.
    \item Minimize the complexity over the choice of $\mathcal{S}$: 
    \be\label{eq:PLC3}
    \Co(\mathcal{R}) \sim \min\limits_{\mathcal{S}}\max\limits_{i<j} \left\lbrace \, \exp\left( \dfrac{S_{\text{gen}}(X_i^b)-S_{\text{gen}}(X_j)}{2} \right) \right\rbrace \;.
    \ee  
\end{enumerate}
This last step can be viewed as a slight generalization of the conjecture \eqref{eq:PLC2}, for cases in which multiple choices of $\mathcal{S}$ exist on $N$.

\begin{figure}[h]
    \centering
    \includegraphics[width = .45\textwidth]{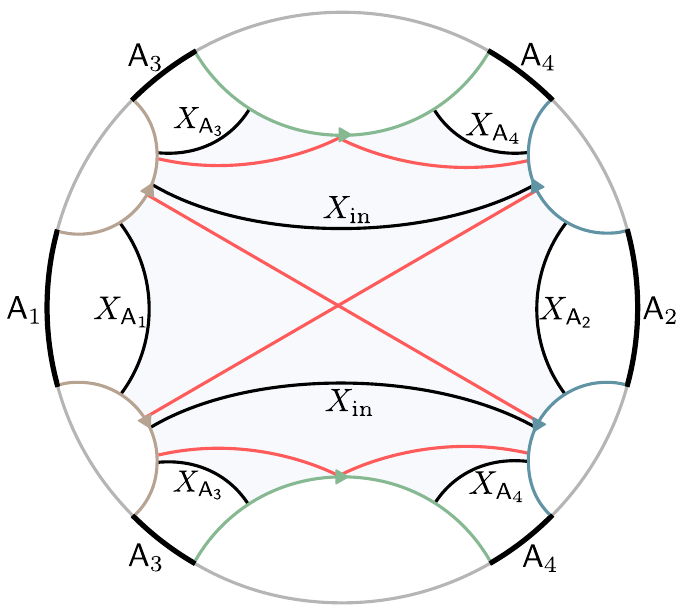}
    \hspace{1cm}
    \includegraphics[width = .45\textwidth]{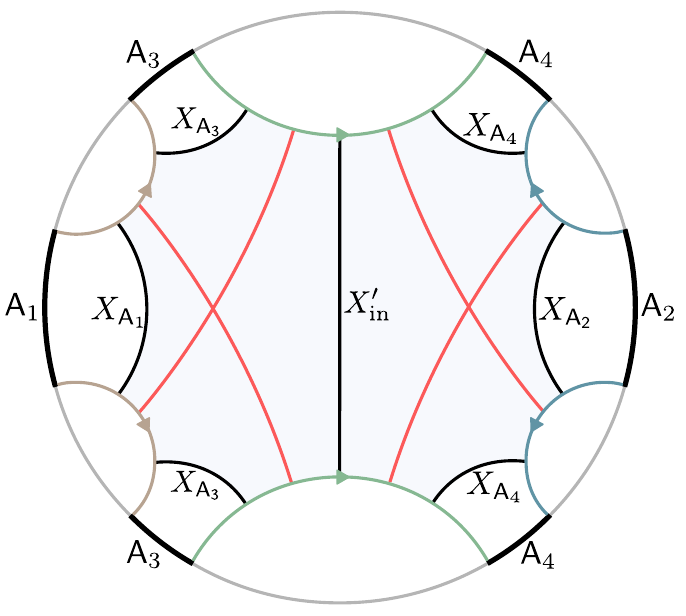}
    \caption{The fundamental domain of the Fuchsian group $\Gamma$ uniformizing the four-boundary wormhole $\ket{\Sigma^4_0}$. The additional two moduli can be taken as the length and twist parameter of e.g. $X_{\text{in}}$ or $X_{\text{in}}'$. On the left, given this choice of interior minimal surface, the python for $\R_1\R_2$, which includes two lunches. The assumption is that $S_{\text{gen}}(X_{\text{in}}) > S_{\text{gen}}(X_{\R_1}) + S_{\text{gen}}(X_{\R_2})$, so that the entanglement wedge of $\R_1\R_2$ contains the whole lunch. On the right, a different choice of closed geodesic $X_{\text{in}}'$ (which intersects $X_{\text{in}}$) corresponds to a different python. The complexity will be given by the minimal complexity amongst all the choices of minimal surfaces in the interior. }
    \label{fig:4bdy}
\end{figure}

Consider the case $n=4$ for concreteness, where the fundamental domain of $\Gamma$ is presented in Fig.\ \ref{fig:4bdy}. In this case, there are $2$ additional moduli that determine the state, aside from the horizon lengts $L_{i}$, for $i=1,2,3,4$. Using Fenchel-Nielsen coordinates, these moduli can be chosen to be the length and twist parameters of an additional closed minimal surface in the black hole interior, such as e.g.\ $X_{\text{in}}$ or $X_{\text{in}}'$ in Fig.\ \ref{fig:4bdy}. Note that these two locally minimal surfaces intersect each other.

Following the steps of the general procedure, in Fig.\ \ref{fig:4bdy} we represent two choices of $\mathcal{S}$ for region $\R = \R_1\R_2$, namely the two corresponding to $X_{\text{in}}$ or $X_{\text{in}}'$. Evaluating the reconstruction complexity \eqref{eq:PLC3} requires minimizing over the choices of this minimal surface in the black hole interior. The optimal choice of $\mathcal{S}$ and the complexity will depend on the moduli on a complicated way, and we will not attempt to quantify the different regimes here.

\section{Discussion \& outlook}
\label{sec:conclusions}

In this paper we have provided an extensive study of \textit{bulges}, extremal surfaces of Morse index 1, found in time reflection-symmetric python geometries. We have related the definition of the bulge to Almgren-Pitts min-max theory, and studied its topological and geometric properties for a variety of time reflection-symmetric states of the holographic system. Our results are potentially useful for testing the python's lunch conjecture, if properties of the complexity of reconstruction can be independently studied. These results also generate some open questions that we list here, in order to structure some possible avenues of future research:

\paragraph{Are black brane interiors really simple?}

We have found that the bulge generally breaks the spatial isometries of the python, and that among other examples, this is particularly relevant for the reconstruction of extended boundary systems. In particular, for black brane interiors the bulge approximately coincides with the constriction exept on a finite region, and the log-complexity of reconstruction predicted by \eqref{eq:PLC} is not extensive in the size of the system. At face value, this seems to suggest that there exists a ``simple'' way to reconstruct these lunches with all the accessible entropy of the boundary system. Following the tensor network/geometry intuition, one such way would be to apply suitable unitaries which break planar symmetry, following a minimax foliation of the python that contains the bulge.

Another situation where naively simple interiors were present was originally faced in \cite{Engelhardt:2021qjs} when considering an equilibrated AdS black hole formed by the collapse of matter.\footnote{We thank Netta Engelhardt for pointing out this similarity to us.} In that case, the reconstruction of the black hole interior seems simple a priori since the spacetime lacks a non-minimal QES. The lesson of \cite{Engelhardt:2021qjs} is that there is an implicit choice of {\it code subspace}, namely the bulk Hilbert space of the quantum fields, in the definition of the bulk-to-boundary map $V$. Given a code subspace, to say that $V$ is simple, one needs to make sure that there are no pythons for {\it any} state of the code subspace. If one picks a late time slice $\Sigma_t$ of the black hole interior, the full ``bulk effective field theory'' code subspace on $\Sigma_t$ is large, since the volume of $\Sigma_t$ scales with the black hole entropy $S$. Given a general excited state of this code subspace on $\Sigma_t$, the state will backreact substantially a scrambling time towards the past, and generate a past singularity. Moreover, the entanglement entropy of the bulk state might need to be considered in the full-fledged QES prescription. In general, these effects create a non-minimal QES in the backreacted spacetime. Therefore, reconstructing large code subspaces in $\Sigma_t$ is exponentially hard. Roughly speaking, the original reconstruction seemed simple because one was implicitly restricting to a small code subspace of states in $\Sigma_t$, namely those that escape the interior under time evolution towards the past.

With this in mind, we come back to our black brane system. The difference with the case in \cite{Engelhardt:2021qjs} is that for the black brane there is already a classical python, but the complexity density of reconstructing the lunch does not scale extensively with the entropy of the boundary system. One might wonder whether there is an impicit choice of a small code subspace in our case as well, that renders the interior reconstruction simple by the same reason as for the black hole. However, it is possible to see that this is not the case. In our $2+1$ dimensional example, we can consider modest but large code subspaces with extensive entropy, consisting of a single degree of freedom per position $x = n x_0$, for $n\in \mathbf{Z}$, at some fixed radial distance in the black brane interior, where $x_0 \sim O(\ell_{\rm AdS})$. Consider a generic state of this code subspace. The dilute backreaction of this state, together with its dilute entropy density, will modify the bulge locally. Since the original bulge is a classically extremal surface, its total change in generalized entropy is controlled by the entropy of the bulk state at leading order, $\delta S_{\text{gen}}(X^b_\R) \sim S_{\text{bulk}}(\rho_{\text{out}}) \sim O(\Lambda_{\text{IR}}^{-1})$. This provides a log-complexity that scales with the log dimension of the code subspace, which by assumption is extensive in the size of the transverse dimension. Nevertheless, the dimension of the code subspace does not scale with $N^2$, unlike the black brane entropy, which is $S_{\text{gen}}(X^c_\R) \sim O(N^2 \Lambda_{\text{IR}}^{-1})$. Thus, naively, the problem still remains to understand why, for extensive code subspaces, the log-complexity to reconstruct them seems not to scale with the entropy of the system. Moreover, it is not obvious that there are simple states at all, given that the states we consider are not formed by collapse. It would be interesting to understand these issues better.

\paragraph{Bulges and compact dimensions}

 Given that the bulge is sensitive to the internal manifold $Y$, its generalized entropy will contain information of the specific holographic system (via e.g.\ its internal global symmetries), which goes beyond purely spatial correlations of the ground state. From the analogy between geometry and tensor networks, the optimal way to decode the local physics of the entanglement wedge from $\R$ will involve the higher-dimensional bulge $X^b_\R$, and will thus necessarily include a non-trivial foliation of the internal manifold. This effect is relevant in all known microscopic constructions of AdS/CFT due the lack of scale separation between the scale of $Y$ and the AdS scale. At the same time, such an observation poses a challenge to the tensor network models of python geometries in the ground state, which are constructed solely from the RT formula, with no specific dynamical input, which makes them insensitive to the internal manifold. It would be interesting to see whether tensor network toy models can incorporate this effect.  

\paragraph{Relation to entwinement and matrix space entanglement}

In \cite{Balasubramanian:2014sra,Balasubramanian:2016xho,Balasubramanian:2018ajb,Craps:2022pke} the area of the bulge $X_\R^b$ in the orbiforld AdS$_3/\mathbf{Z}_n$ was interpreted as measuring ``entwinement'', a quantity that emerges from the orbifold description of the theory, with boundary interpretations offered in \cite{Lunin:2000yv,Martinec:2001cf,Martinec:2002xq,Balasubramanian:2005qu,Craps:2022pke}. In section \ref{sec:excitedstates}, we have pointed out that this same surface is associated with the bulge of the python, which measures some notion of the 
complexity of reconstructing the entanglement shadow with access only to $\R$. It is possible that these two interpretations are connected in some way. 

On the other hand, in the context of LLM geometries, one motivation for \cite{Balasubramanian:2017hgy}\ was to better understand surfaces that bisected the $ S^5$ (or related interior geometries) in the AdS duals of ${\cal N}= 4$ $U(N)$ super-Yang Mills theory on the Coulomb branch. These were studied in \cite{Mollabashi:2014qfa,Karch:2014pma} who argued compellingly that such surfaces meaured entanglement between matrix degrees of freedom in the field theory. However, \cite{Balasubramanian:2017hgy}\ showed that at the origin of the Coulomb branch, the surface bisecting the $ S^5$ was not minimal; rather, if we cut off AdS${}_5$ at large radius, dual to a UV cutoff in the field theory, the minimal surface hugs the cutoff. This points to the surfaces studied in \cite{Mollabashi:2014qfa,Karch:2014pma} as having an intepretation as some kind of complexity, perhaps related to the matrix degrees of freedom of the theory.

\paragraph{Multi-boundary wormhole states}

For multi-boundary wormhole states of the black hole which include pythons, we have found that, considering regions $\R$ comprised of multiple boundaries, the complexity of reconstructing the lunch plateaus after some number of connected components have been included in $\R$. This feature is closely related to the non-extensivity of the log-complexity for extended systems such as black branes. In this case, certain connected components of the bulge and constriction coincide exactly. 

Given a multi-boundary wormhole state, this effect implies that all the quantum information of the code subspace in the shared interior of the multiple black holes is encoded via $V$ in a subset of them; applying non-trivial unitaries to the rest results in a reconstruction which is not optimal. We note that this feature, and the possibility of applying the geometric PLC to multi-boundary regions $\R$, should be captured by generic tensor network models of the multi-boundary python, under the assumption that the individual tensors that constitute the geometry appear uncorrelated, at least approximately. In fact, for the multi-boundary wormhole states $\ket{\Sigma^n_0}$ that we have analyzed, a Haar random state model was originally proposed in \cite{Balasubramanian:2014hda}, which captured the mutual information of these states. A finer model of the states $\ket{\Sigma^n_0}$ is to consider a random tensor network model of the geometry, with tensors of large bond dimension. Such a model directly captures the physics of the RT formula \cite{Hayden:2016cfa}, and therefore the multipartite entanglement structure. Moreover, in random tensor network models with uncorrelated local tensors, it seems reasonable to expect that the way to implement post-selection unitarily is generically via brute-force Grover search locally, so they will satisfy the assumptions of the PLC and will reproduce the geometric features that we have studied. 

Without knowledge of how the multi-boundary states with shared interiors are prepared unitarily, the genericity assumption seems reasonable. However, this feature is in contrast with ``simple'' tensor network states such as the ones motivating the conjecture \cite{Brown:2019rox}, where the bulge arises for proper subregions $\R$ as an artifact of the coupling of $\R$ to $\bar{\R}$ in the unitary time evolution that drives the full system together (in the case of \cite{Brown:2019rox} $\R$ is the early Hawking radiation and $\bar{\R}$ is the black hole). For such states, the tensors in the tensor network are correlated and the bulge dissapears once $\R$ is taken to be the full system. { It would be interesting to test this prediction of the conjecture and understand if and why geometric bulges cannot form dynamically in a simple way, when allowing operators which couple the different boundaries.}

\paragraph{Time dependence}

In our analysis, we have restricted ourselves to time reflection-symmetric states, where all of the extremal surfaces lie on the time reflection-symmetric Cauchy slice $\Sigma$. Situations in which the time-reflection symmetry is spontaneously broken by the bulge and other locally minimal surfaces have been reported in \cite{Engelhardt:2023bpv} for specific spherically symmetric initial data in near-extremal black hole interiors. The specific data has been constructed using the two-dimensional description of these systems given by JT gravity with additional matter fields coupled to the metric. In all of the examples that we have analyzed in this paper, however, we do not expect such an effect. An interesting open problem that we leave for future work is to investigate the nature of more general extremal surfaces, such as the surfaces called \textit{bounces} (cf. \cite{Engelhardt:2023bpv}), that are expected to arise in these situations.

\section*{Acknowledgments}

We thank Chris Akers, \AA smund Folkestad, Brianna Grado-White, Guglielmo Grimaldi, Veronika Hubeny, Dominik Neuenfeld, Mukund Rangamani, and especially Netta Engelhardt and Brian Swingle for useful conversations. We are grateful to the long term workshop YITP-T-23-01 held at YITP, Kyoto University, as well as to the Simons Foundation, the Institute for Advanced Study, and the Centro de Ciencias de Benasque Pedro Pascual, where part of this work was completed. This work was supported in part by the Department of Energy through awards DE-SC0009986 and QuantISED DE-SC0020360, and partly by the Simons Foundation through the \emph{It from Qubit} Simons Collaboration.

\appendix

\section{Gauss map trick}
\label{app:B}

It was noted in \cite{Fischer1985} that the index of a complete orientable extremal surface $\Sigma$ in $\mathbf{R}^3$ is only dependent on the Gauss map, $n: \Sigma \to  S^2$ which is defined such that for each $p \in \Sigma$, $n(p)$ is the unit normal vector to $\Sigma$ at $p$. We are interested in extremal surfaces of finite index which is equivalent to the condition of finite total curvature, $\int_{\Sigma}|\kappa_1 \kappa_2| <\infty$ where $\kappa_{1,2}$ are principal curvatures of the extremal surface \cite{Fischer1985}. Further, every extremal surface of this type is conformally equivalent to a compact Riemann surface with punctures, $\bar{\Sigma} \setminus \{p_1,\ldots, p_k\}$ \cite{osserman1964} and thus we can extend the Gauss map $n$ to $\bar{n}: \bar{\Sigma} \setminus \{p_1,\ldots, p_k\} \to  S^2$.

The second order variation of area of $\bar{\Sigma}$ is (we are going to omit writing the punctured points explicitly)
\be\label{eq:variationSigmaBar}
\delta^{(2)}\text{Area}(\bar{\Sigma},h) = \dfrac{1}{2}\int_{\bar{\Sigma}} \sqrt{h}\, \left(h^{ij}\partial_i \eta \partial_j \eta  - K_{ij}K^{ij}\eta^2 \right) \,.
\ee
The quadratic form \ref{eq:variationSigmaBar} is invariant under conformal variations of the type $\tilde{h}_{ij} = e^{2\omega}h_{ij}$. We get the following transformations under Weyl scaling:
\begin{gather}
    \sqrt{\tilde{h}} = e^{2 \omega} \sqrt{h}\\
    \tilde{h}^{ij} = e^{-2 \omega} h^{ij}\\
    \tilde{K}_{ij} = e^{\omega} K_{ij} 
\end{gather}
Hence,
\be \label{eq:conformalInvariant}
\delta^{(2)}\text{Area}(\bar{\Sigma},\tilde{h}) = \delta^{(2)}\text{Area}(\bar{\Sigma},h) \, .
\ee

To illuminate the connection between the (extended) Gauss map and the index, we note that
\begin{align*}
K_{ij}K^{ij} &= K_{ij}h^{il}h^{jm}K_{lm}\\
&= \Tr(P^2)\\
&= -2 \kappa
\end{align*}
where $P$ is the shape operator and $\kappa$ is the Gaussian curvature. We also have a linear map $d\bar{n}: T_p \bar{\Sigma} \to T_{\bar{n}(p)} S^2$ with the property that its determinant is the Gaussian curvature, $\kappa$. Therefore, we get
\be
K_{ij}K^{ij} = -2 \kappa = -2 \det(d\bar{n}) \, .
\ee

Since the quadratic form \ref{eq:variationSigmaBar} is conformally invariant as demonstrated by \ref{eq:conformalInvariant}, we can choose the rescaled metric $\tilde{h}_{ij} = -\kappa h_{ij}$. Then \ref{eq:variationSigmaBar} simplifies to
\be\label{eq:quadraticform}
\delta^{(2)}\text{Area}(\bar{\Sigma},\tilde{h}) =- \dfrac{1}{2}\int_{\bar{\Sigma}} \sqrt{\tilde{h}}\, \eta\left(\nabla^2  + 2 \right)\eta \,.
\ee
In fact, the metric $\tilde{h}_{ij}$ is the pullback under the Gauss map\footnote{The linear map $d\bar{n}: T_p \bar{\Sigma} \to T_{\bar{n}(p)} S^2$ can be extended to a \textit{Weingarten} map $W : T_p \bar{\Sigma} \to T_p \bar{\Sigma}$ and technically, the metric $\tilde{h}_{ij}$ is a pullback under the Weingarten map.} $\bar{n}$ and hence, $\tilde{h}_{ij} = -\det(d\bar{n}) h_{ij}$.

We are interested in finding the number of negative modes of the quadratic form \ref{eq:variationSigmaBar}. The most natural vector space on which the quadratic form $A$ can act is $L^2_h(\Sigma)$ with inner product $\langle\phi,\psi\rangle = \int_{\Sigma} \sqrt{h}\phi^*\psi$. The index of $A$ is defined as the dimensionality of the largest subspace of $L^2_h(\Sigma)$ on which $A$ is negative definite. Note that the inner product is not preserved under the action of Gauss map but it is true that $L^2_h(\Sigma) \subset L^2_{\tilde{h}}(\bar{\Sigma})$ since the Gaussian curvature is bounded from above. A rigorous proof of index$(\Sigma)$ = index$(\bar{\Sigma})$ involves constructing a basis of $L^2_h(\Sigma)$ using a basis of $L^2_{\tilde{h}}(\bar{\Sigma})$ and showing that the span of either set of basis vectors is the same. This was shown in \cite{Fischer1985} with the assumption that $\tilde{h}$ is a smooth metric on $\bar{\Sigma}$ (including the punctures).

We will give a variant of the proof in \cite{Fischer1985} in case of a catenoid. We have $\Sigma = $ catenoid, $\bar{\Sigma} =  S^2\setminus\{\text{N,S}\}$ where N,S are the two poles, $\bar{n}:  S^2\setminus\{\text{N,S}\} \to  S^2$. The induced metric $\tilde{h}_{ij}$ is smooth at punctures and hence the quadractic form \ref{eq:quadraticform} can be analyzed on $ S^2$ with the round metric. The eigenvalues are given by $\lambda_{l,n} = l(l+1)-2$ and therefore the index($\bar{\Sigma}$)$= 1$. Since $L^2_h(\Sigma) \subset L^2_{\tilde{h}}(\bar{\Sigma})$, index$(\Sigma) \leq 1$. To show that it is equal to 1, we need to find $\psi \in L^2_h(\Sigma)$ such that $\delta^{(2)}\text{Area}(\bar{\Sigma},\tilde{h}) < 0$. To ensure that $\psi$ is in  $L^2_h(\Sigma)$, we need to put regularity condition that $\psi \to 0$ as we approach the puncture. Let us consider the standard round metric, $d\tilde{s}^2 = d\theta^2 + \sin^2 \theta d\phi^2$ and a function $\psi$ given by
\be
\psi = \begin{cases}
    \frac{\theta - \theta_p}{\theta_0}&\theta- \theta_p<\theta_0\\
    1&\theta- \theta_p \geq \theta_0
\end{cases}
\ee
near each puncture $\theta_p$ and $\theta_0 \ll 1$. Then \ref{eq:quadraticform} becomes
\be
\delta^{(2)}\text{Area} = \frac{\pi}{2}p - 4\pi
\ee
where $p$ is the total number of punctures. In case of a catenoid, $p=2$ and hence $\delta^{(2)} \text{Area} <0$.

\section{Weierstrass--Enneper representation}
\label{app:WE-rep}

The Weierstrass--Enneper representation is a convenient way of parametrizing extremal surfaces in $\mathbf{R}^3$. Let $f$ be an analytic function and $g$ a meromorphic function on some domain in $\mathbf{C}$, such that $fg^2$ is analytic. This will furnish an extremal surface in $\mathbf{R}^3$ with embedding coordinates given by
\begin{gather}
    x = \Re \left( \int_0^z dz'(1-g(z')^2)f(z')  \right)\\
    y = \Re \left( \mathbf{i}\int_0^z dz'(1+g(z')^2)f(z')  \right)\\
    z = \Re \left( 2\int_0^z dz'f(z')g(z') \right).
\end{gather}
In fact, any nonplanar extremal surface in $\mathbf{R}^3$ can be represented by the above parametrization.

For a singly periodic Scherk surface, the domain is the unit disk, and
\begin{equation}
    f(z) = \frac{4}{(z^2-z_0^2)(z^2-\bar{z}_0^2)}\;,\qquad
    g(z) = \mathbf{i} z\;,
\end{equation}
where $z_0 = e^{\mathbf{i}\phi}$ with $\phi$ being the half angle between the planes. Therefore, the parametric form for the Scherk surface is:
\begin{gather}
    x(r, \theta) = \frac{1}{2 \sin \phi}\left\{ \ln \left(\frac{1+r^2+2r\cos(\theta + \phi)}{1+r^2-2r\cos(\theta+\phi)}\right) - \ln \left(\frac{1+r^2+2r\cos(\theta - \phi)}{1+r^2-2r\cos(\theta-\phi)}\right)\right\}\\
    y(r, \theta) = \frac{1}{2 \cos \phi}\left\{ \ln \left(\frac{1+r^2+2r\cos(\theta + \phi)}{1+r^2-2r\cos(\theta+\phi)}\right)+\ln \left(\frac{1+r^2+2r\cos(\theta - \phi)}{1+r^2-2r\cos(\theta-\phi)}\right) \right\}\\
    z(r, \theta) = \frac{1}{\cos \phi \sin \phi} \left\{ \arctan \left( \frac{\sin 4\phi -2r^2 \cos 2\theta \sin 2\phi}{\cos 4\phi +r^4 -2r^2 \cos 2\theta \cos 2\phi} \right)  -4\phi \right\}
\end{gather}
where $r \in (0,1)$ and $\theta \in [0,2 \pi)$.

\bibliographystyle{utphys}
\bibliography{bibliography}

\end{document}